%% file: main.tex
\documentclass{article}

\usepackage[affil-it]{authblk}
\usepackage[dvipsnames]{xcolor}
\usepackage{amsfonts}
\usepackage{amsmath,amsthm,amssymb,dsfont}

\usepackage{enumerate}
\usepackage{graphicx}	
\usepackage{subcaption}
\usepackage[margin=3cm]{geometry}
\usepackage{url}
\usepackage{todonotes}
\usepackage{bbm}

\usepackage{tikz}
\usepackage{circuitikz}

\usetikzlibrary{arrows.meta} 
\tikzset{ meter/.append style={ draw, inner sep=10, rectangle, font=\vphantom{A}, minimum width=30, line width=.5, path picture={ \draw[black] ([shift={(.1,.3)}]path picture bounding box.south west) to[bend left=50] ([shift={(-.1,.3)}]path picture bounding box.south east); \draw[black,-{Latex[scale=0.6]}] ([shift={(0,.1)}]path picture bounding box.south) -- ([shift={(.3,-.1)}]path picture bounding box.north); } } }

\usepackage{pifont}
\usepackage{multirow}
\usepackage{makecell}

\usepackage{epsfig}
\usetikzlibrary{shapes.symbols,patterns} % for source symbols
\usepackage{pgfplots}
\pgfplotsset{compat=1.10}
\usepgfplotslibrary{fillbetween}
\usetikzlibrary {decorations.pathmorphing, decorations.pathreplacing, decorations.shapes}

\definecolor{linkblue}{HTML}{001487}
\usepackage{hyperref}
\hypersetup{colorlinks=true,citecolor=linkblue,linkcolor=linkblue,filecolor=linkblue,urlcolor=linkblue,breaklinks=true}

\usepackage{nicefrac}
\usepackage{mathtools}

\usepackage{thmtools} % Load this to fix cleveref issues
\hypersetup{hypertexnames=false}

\usepackage{algorithm}
\usepackage{algorithmic}

\usepackage{mdframed}
\usepackage{aligned-overset}
\usepackage{circuitikz}

\theoremstyle{plain}
\newtheorem{theorem}{Theorem}[section]
\newtheorem{lemma}[theorem]{Lemma}

\theoremstyle{definition}

\newtheorem{remark}[theorem]{Remark}

\newtheorem{non-example}[theorem]{Non-example}

\newcommand{\ketbra}[2]{|#1\rangle\!\langle#2|}

\DeclareRobustCommand{\abbrevcrefs}{%
\Crefname{theorem}{Thm.}{Thms.}%
\Crefname{corollary}{Cor.}{Cors.}%
\Crefname{lemma}{Lem.}{Lems.}%
\Crefname{remark}{Rmk.}{Rmks.}%
\Crefname{proposition}{Prop.}{Props.}%
\Crefname{equation}{Eq.}{Eqs.}%
\Crefname{example}{Ex.}{Exs.}%
}

\DeclareRobustCommand{\Cshref}[1]{{\abbrevcrefs\Cref{#1}}}

\newcommand*{\ee}{\mathrm{e}}

\newcommand*{\cB}{\mathcal{B}}

\newcommand*{\cH}{\mathcal{H}}

\newcommand*{\cP}{\mathcal{P}}

\newcommand*{\cU}{\mathcal{U}}

\newcommand*{\R}{\mathbb{R}}

\newcommand*{\C}{\mathbb{C}}

\newcommand*{\LOCC}{\mathrm{LOCC}}

\newcommand*{\eps}{\varepsilon}
\newcommand*{\diag}{\mathrm{diag}}

\newcommand*{\id}{\mathds{1}}

\newcommand*{\tr}{\mathrm{tr}}
\newcommand*{\ket}[1]{| #1 \rangle}
\newcommand*{\bra}[1]{\langle #1 |}

\newcommand{\proj}[1]{|#1\rangle\!\langle #1|}
\newcommand*{\braket}[1]{\langle #1 \rangle}

\newcommand*{\adj}[1]{\mathrm{adj}#1}

\newcommand*{\ci}{\mathrm{i}} % imaginary number
\newcommand*{\di}{\mathrm{d}} % integration d

\newcommand{\norm}[1]{\left\lVert#1\right\rVert}

\usepackage{float}
\usepackage[nameinlink,capitalize,noabbrev]{cleveref}

 \allowdisplaybreaks

\definecolor{protocolblue}{HTML}{1F5A94}
\definecolor{protocolorange}{HTML}{C66A2B}
\definecolor{protocolgreen}{HTML}{2F7D68}
\definecolor{protocolgray}{HTML}{68717A}

\title{Essentially optimal gate teleportation} 

    \author{\normalsize Lukas Schmitt$^{1,2}$ and David Sutter$^{2}$}
     \affil{\small $^{1}$Institute for Theoretical Physics, ETH Zurich\\
     $^{2}$IBM Research Europe -- Zurich}
\date{}

\begin{document}

\maketitle

\begin{abstract}
Gate teleportation allows us to implement a nonlocal unitary using local operations, classical communication (LOCC), and a shared entangled state. Known deterministic teleportation protocols consume at least one full ebit and achieve optimal entanglement consumption only for Clifford gates. Here, we present a deterministic LOCC protocol for implementing the two-qubit controlled-phase gate $U_{\phi}=\diag(1,1,1,\ee^{\ci \phi})$ with $\phi \in [0,\pi]$ whose entanglement consumption is close to optimal for every $\phi$. In particular, vanishing rotation angles require vanishing entanglement. 
%We give an information-theoretic proof that the entanglement consumption of the presented protocol cannot be improved substantially.
\end{abstract}

%%%%%%%%%%%%%%%%%%%%%%%%%%%%%%%%%%%%%%%%%%%%%%%%%%%%%%%%%%%%%%%%%%%%%%%%%%%%%%%%%%%%
\section{Introduction}
Gate teleportation provides an operationally meaningful framework for nonlocal quantum computation by implementing a joint unitary on spatially separated systems using pre-shared entanglement, local operations, and classical communication (LOCC). This eliminates the need for a direct coherent interaction between distant systems, as the entanglement can be generated, distributed, and verified in advance. 
Shared pure states with the same entanglement entropy can be (asymptotically) converted into each other under LOCC~\cite{bennett96,MHR02}. For example, $n$ Bell states can be asymptotically converted into $10n$ copies of any pure state containing $0.1$ ebits of entanglement per copy. Consequently, entanglement entropy, measured in ebits, universally quantifies pure-state entanglement resources, independently of the detailed form of both the resource state and the nonlocal unitary to be implemented.
\begin{figure}[!htb]
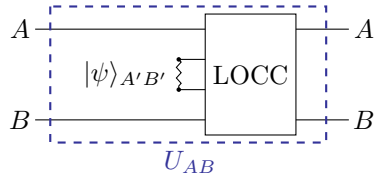

\centering
\include{gate_teleporation_scheme}
\vspace{-10mm}
\caption{Gate teleportation of a nonlocal unitary $U_{AB}$. The unitary is implemented with an entangled resource state $\ket{\psi}_{A'B'}$ together with a LOCC protocol with respect to the bipartition $(A \otimes A'):(B \otimes B')$.}
\label{fig_teleportation}
\end{figure}

A gate teleportation protocol (see~\cref{fig_teleportation}) is optimal if it implements the respective gate exactly and deterministically using a minimal amount of entanglement. 
In general, we can take the resource state to be pure\footnote{This is justified in~\cref{rmk_main}.}, and therefore use entanglement entropy as the relevant metric.
%Without loss of generality, we can take the resource state to be pure\footnote{This is justified in~\cref{rmk_main}.}hence entanglement is fully characterized by the entanglement entropy~\cite{MHR02}. 
We formalize the teleportation cost of a bipartite unitary $U_{AB}$ as\footnote{The unitary channel induced by a unitary $U$ is denoted by $\cU(\cdot)=U(\cdot)U^\dagger$, $\cP$ is an LOCC protocol acting on $AA':BB'$ and $H(A')_\psi$ denotes the entanglement entropy, which is defined as $H(A')_{\psi}:=-\tr[\rho_{A'} \log \rho_{A'}]$ for $\rho_{A'}:=\tr_{B'}[\proj{\psi}].$}
\begin{align} \label{eq_teleportation_problem}
T_C(U_{AB}):=\inf_{\cP \in \LOCC, \ket{\psi}_{A' B'} \in \cH_{ A'  B'}} \big\{ H(A')_{\psi} : \cU_{AB}(\cdot)=  \cP(\cdot, \ket{\psi}_{A'B'}) \big\} \, ,
\end{align}
where $A'$ and $B'$ can have unbounded size. 
For practical purposes, we are also interested in the optimizers in~\cref{eq_teleportation_problem}, which are needed to perform the gate teleportation.

For Clifford gates, optimal gate teleportation is understood~\cite{GC99}, and an optimal resource state is given by the Choi state.\footnote{The Choi state of a unitary $U_{AB}$ is given by $\ket{\psi}_{AA'BB'}:=(U_{AB} \otimes \id_{A'B'})(\ket{\Phi}_{AA'}\otimes \ket{\Phi}_{BB'})$, where $\ket{\Phi}_{AA'}$ denotes a maximally entangled state.} For example, in the case of a $\mathrm{CNOT}$ gate we have $T_C(\mathrm{CNOT})=1$, i.e.~teleportation can be done with one ebit.
Beyond Clifford gates, little is known about optimal gate teleportation. Intuitively, one would hope that a unitary with little entangling power, such as a controlled-phase gate with a small angle, should require only a weakly entangled state. However, existing teleportation protocols cannot do this. For example, in~\cite[Theorem~2]{Eisert00} it was shown that any two-qubit controlled unitary can be teleported with a single ebit. In~\cite{YGC10,co10} further teleportation protocols have been derived that work for specific unitaries, but also consume entire ebits. Using state teleportation~\cite{BBCJPW93}, a local two-qubit unitary implementation is trivially possible with two ebits.

This raises the question whether it is possible to teleport a gate consuming less than an ebit~\cite{Eisert00}. 
In~\cite{SG11}, numerical evidence was presented that this might be the case. However, an actual protocol was still missing.

\paragraph{Results:}\label{par:results} Let $U_{\phi}=\diag(1,1,1,\ee^{\ci \phi})$ be the two-qubit controlled-phase gate with $\phi \in [0,\pi]$. We prove the following achievability and converse bounds.
\begin{enumerate}[(a)]
\item \textbf{Achievability:} We present a deterministic gate teleportation protocol for $U_{\phi}$ given in~\cref{fig_new_teleportation}, which uses a resource state $\ket{\psi_\phi}_{A'B'}$ with Schmidt rank $3$ that satisfies 
\begin{align}
H(A')_{\psi_{\phi}} =  h \Big(\frac{1}{1 + \sin(\phi/2)} \Big) + \frac{\sin(\phi/2)}{1+\sin(\phi/2)} \, , \label{eq_achievablity}
\end{align}
where $h(x):=-x \log(x) - (1-x) \log(1-x)$ is the binary entropy.\footnote{Logarithms are taken with base $2$.}
The precise form of the resource state, unitaries, measurements, and the proof of correctness are given in~\cref{sec_pf_achievability}.
\begin{figure}[!htb]
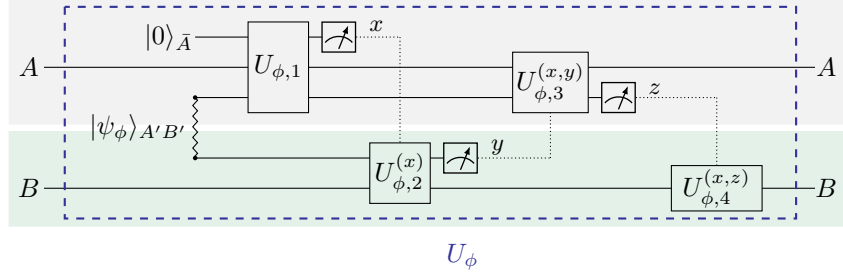

\centering
\include{new_gate_teleportation_scheme}
\vspace{-10mm}
\caption{Gate teleportation protocol for the controlled-phase gate $U_{\phi}=\diag(1,1,1,\ee^{\ci \phi})$ with $\phi \in [0,\pi]$ using the entangled resource state $\ket{\psi_{\phi}}_{A'B'}$. Systems $A$ and $B$ denote the input qubits. Dotted lines depict classical communication.}
\label{fig_new_teleportation}
\end{figure}

\item \textbf{Converse:} We show that any protocol that can teleport $U_{\phi}$ must consume a resource state $\ket{\psi}_{A'B'}$ whose entanglement entropy satisfies
\begin{align}
H(A')_\psi \geq h \Big(\frac{1}{1 + \sin(\phi/2)} \Big)\, . \label{eq_converse}
\end{align}
The proof is given in~\cref{sec_pf_converse}.
% Unlike for Clifford gates, the Choi state of $U_{\phi}$ cannot itself be used as the resource state for exact teleporation when $\phi \in (0,\pi)$.\footnote{Recall that the Choi state $\ket{\psi}_{AA'BB'}:=(U_{\phi} \otimes \id_{A'B'})(\ket{\Phi}_{AA'}\otimes \ket{\Phi}_{BB'})$ has entanglement entropy $H(AA')_{\psi} = h(\sin^2(\phi/4))$ which is strictly below $h (\frac{1}{1 + \sin(\phi/2)})$ for all $\phi \in (0,\pi)$. }
When assuming that the resource state has a fixed finite Schmidt rank, a substantially tighter converse bound can be derived. This is discussed in~\cref{app_improved_converse}.
\end{enumerate}

\noindent It is known~\cite[Theorem~2]{Eisert00} that $U_{\phi}$ can be teleported with one ebit. Hence, whenever the protocol from~\cref{fig_new_teleportation} requires more than one ebit we can switch to the protocol from~\cite[Theorem~2]{Eisert00}. Therefore, we have 
\begin{align} \label{eq_summary_results}
 h\Big(\frac{1}{1 + \sin(\phi/2)} \Big)
\overset{\textnormal{\Cshref{eq_converse}}}{\leq} T_C(U_{\phi})  
\overset{\textnormal{\Cshref{eq_achievablity}}}{\leq} \min\Big \{1,  h \Big(\frac{1}{1 + \sin(\phi/2)} \Big) + \frac{\sin(\phi/2)}{1+\sin(\phi/2)} \Big \} \, .
\end{align}
\Cref{fig_results_curve} plots these bounds for $\phi \in [0,\pi]$.
\begin{figure}[!htb]
\centering
\input{Fig_results}
\caption{Entanglement entropy of the resource state for teleporting a controlled-phase gate $U_{\phi}=\diag(1,1,1,\ee^{\ci \phi})$ of angle $\phi \in [0,\pi]$. The achievability bound shows the specific teleportation protocol explained in~\cref{fig_new_teleportation}. Below the dashed converse bound (in black) no teleportation is possible. If we limit the Schmidt rank of the resource state to 3, a stronger converse bound holds (in dotted blue).}
\label{fig_results_curve}
\end{figure}
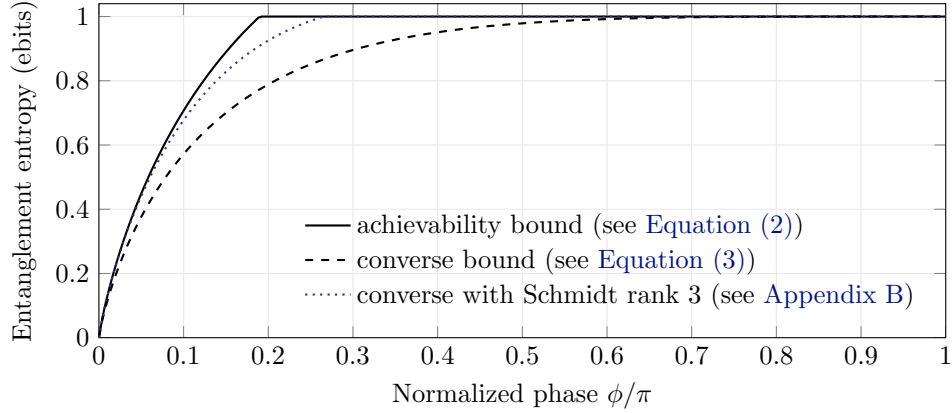

\begin{remark}[Beyond controlled-phase gates] \label{rmk_unitaries}
The teleportation protocol from~\cref{fig_new_teleportation} for the controlled-phase gate $U_{\phi}=\diag(1,1,1,\ee^{\ci \phi})$ extends to broader classes of two-qubit unitaries.
\begin{enumerate}[(i)]
\item \textbf{Arbitrary controlled two-qubit unitaries:} Let $\bar U = \proj{0} \otimes V_0 + \proj{1} \otimes V_1$ be an arbitrary controlled two-qubit gate where $V_0$ and $V_1$ are arbitrary single-qubit unitaries. Diagonalizing $V_0^\dagger V_1 = S\, \diag(\ee^{\ci \alpha}, \ee^{\ci \beta}) S^\dagger$ and setting $D_{\alpha}=\diag(1,\ee^{\ci \alpha})$ we have the identity
\begin{align} \label{eq_teleporting_general_gate}
\bar U = (D_{\alpha} \otimes V_0 S) U_{\beta - \alpha} (\id \otimes S^\dagger) \, . 
\end{align}
Since single qubit gates are local,~\cref{eq_teleporting_general_gate} shows that teleporting $\bar U$ is equivalent to the teleportation of $U_{\beta - \alpha}$.\footnote{Note that $U_{\phi}$ is locally equivalent to $U_{-\phi}$.} 

A similar argument also applies for two-qubit rotations of the form $U_{\sigma}(\theta):=\exp( - \ci \frac{\theta}{2} \sigma \otimes \sigma)$ where $\sigma \in \{X,Y,Z\}$. For $|\theta| \leq \pi/2$, this unitary is locally equivalent to $U_{2|\theta|}$.

\item \textbf{Arbitrary two-qubit unitaries:} Any two-qubit unitary has a KAK decomposition
\begin{align}
U_{AB} = (V_1 \otimes V_2) \exp\Big(- \frac{\ci}{2}(\theta_X X \otimes X + \theta_Y  Y \otimes Y + \theta_Z Z \otimes Z) \Big) (V_3 \otimes V_4) \, ,
\end{align}
for single-qubit unitaries $V_1,V_2,V_3,V_4$, and $|\theta_\sigma|\leq \pi/2$ for $\sigma \in \{X,Y,Z\}$.
Since the three Pauli products commute, $U_{AB}$ is locally equivalent to $U_{X}(\theta_X) U_{Y}(\theta_Y) U_{Z}(\theta_Z)$. Hence, 
\begin{align}
T_C(U_{AB}) 
&= T_C\big(U_{X}(\theta_X) U_{Y}(\theta_Y) U_{Z}(\theta_Z)\big) \\
&\leq T_C\big(U_{X}(\theta_X) \big) +  T_C\big(U_{Y}(\theta_Y) \big) +  T_C\big(U_{Z}(\theta_Z) \big) \label{eq_step_suboptimal}\\
& = T_C(U_{2|\theta_X|}) + T_C(U_{2|\theta_Y|}) +T_C(U_{2|\theta_Z|}) \, . \label{eq_general_2qubit}
\end{align}
The protocol from~\cref{fig_new_teleportation} can be used to teleport $U_{AB}$ (by separately teleporting $U_{X}(\theta_X)$, $U_{Y}(\theta_Y)$, and $U_{Z}(\theta_Z)$) consuming
\begin{align}
\sum_{\sigma \in \{X,Y,Z\}} \left( h \Big(\frac{1}{1 + \sin(|\theta_\sigma|)} \Big) + \frac{\sin(|\theta_\sigma|)}{1+\sin(|\theta_\sigma|)} \right)
\end{align}
ebits. We emphasize that this teleportation protocol may be suboptimal due to the step in~\cref{eq_step_suboptimal}, where we treat the nonlocal unitary $U_{X}(\theta_X) U_{Y}(\theta_Y) U_{Z}(\theta_Z)$ as three individual unitaries (instead of a single unitary). However, it may still be more efficient than the trivial gate teleportation protocol that consumes two ebits.

\end{enumerate}
\end{remark}

\begin{remark}[Properties of the resource state] \label{rmk_main} \
\begin{enumerate}[(i)]
\item \textbf{The resource state is pure:} Assume that $\cP(\cdot,\rho) \in \LOCC$ is a teleportation protocol that, together with a mixed resource state $\rho_{A'B'}$ implements a unitary $U_{AB}$. We can write $\rho_{A'B'} = \sum_i p_i \proj{\psi_i}_{A'B'}$, where we choose a decomposition that achieves the entanglement of formation~\cite{BDSW96}. By linearity of the teleportation protocol we have
\begin{align}
\cU(\cdot) = \cP(\cdot,\rho) = \sum_{i} p_i \cP(\cdot,\ket{\psi_i}) \, .
\end{align}
Unitary channels are the extreme points of the set of all completely positive and trace-preserving maps~\cite[Theorem~5]{choi75}. Hence, for all $i$ such that $p_i>0$ we must have $\cP(\cdot,\ket{\psi_i}) = \cU(\cdot)$. Thus, we have shown that we can also teleport with the pure resource state $\ket{\psi_i}$.
In addition, since we are looking for the pure resource state $\ket{\psi_i}$ that has minimal entanglement, we see that at least one component satisfies $H(A')_{\psi_i} \leq E_F(\rho)$, where $E_F$ denotes the entanglement of formation~\cite{BDSW96}. Thus, the entanglement of a pure resource state is also smaller than for mixed states.

\item \textbf{Schmidt rank at least three:} For $\phi \in (0,\pi]$ the operator Schmidt rank of $U_{\phi}$ is two, and a pure resource state of the same Schmidt rank must be maximally entangled in any deterministic LOCC implementation~\cite{SG11,STM11}. Thus, any teleportation protocol that consumes less than an ebit must exploit a resource state with Schmidt rank at least three.
\end{enumerate}
\end{remark}

\paragraph{Application:} To illustrate the usefulness of our teleportation scheme, consider the implementation of the quantum Fourier transform (QFT)~\cite[Section~5]{nielsenChuang_book} in a distributed setting.
Suppose that we want to implement a (bit-reversed) QFT on $2n$ qubits where Alice holds the first $n$ qubits and Bob the last $n$ qubits. The task is to implement this QFT with an entangled resource state and LOCC operations. Construction~\cite[Theorem~2]{Eisert00}, which teleports every nonlocal controlled unitary consuming one ebit, uses $n^2$ ebits in total. In contrast, with the teleportation protocol presented in~\cref{fig_new_teleportation}, we can implement the QFT with a constant number of $12.1869$ ebits for any $n$. This is explained in detail in~\cref{sec_QFT}. 

%%%%%%%%%%%%%%%%%%%%%%%%%%%%%%%%%%%%%%%%%%%%%%%%%%%%%%%%%%%%%%%%%%%%%%%%%%%%%%%%%%%
%%%%%%%%%%%%%%%%%%%%%%%%%%%%%%%%%%%%%%%%%%%%%%%%%%%%%%%%%%%%%%%%%%%%%%%%%%%%%%%%%%%
\section{Proof of the achievability result} 
\label{sec_pf_achievability}
\subsection{Protocol}
The case $\phi=0$ is trivial and we will assume $\phi \in (0,\pi]$. 
As explained in~\cref{rmk_main} we need to use a resource state with Schmidt rank at least three to consume strictly less entanglement than one ebit.
Consider
\begin{align} \label{eq_form_resource_state}
    \ket{\psi_p}_{A'B'} = \sqrt{p_0} \ket{0,0}_{A'B'} + \sqrt{p_1} \ket{1,1}_{A'B'} + \sqrt{p_2} \ket{2,2}_{A' B'} \, ,
\end{align}
where $p=(p_0,p_1,p_2) \in \R^3_+$ satisfies $p_0 + p_1 + p_2=1$, $p_0\geq1/2\geq p_1\geq p_2>0$ and 
\begin{equation}\label{eq_feasibility}
    \sin^2(\phi/2)(1-2p_1)(1-2p_2)\leq 4p_1p_2 \, .
\end{equation}
If $p_0<1/2$, the entanglement entropy would be greater than 1 (see~\cref{lem_tailbound}).
We show that the protocol in~\cref{fig_new_teleportation} implements $U_\phi = \diag(1,1,1,\ee^{\ci \phi})$ for every resource state that satisfies~\cref{eq_feasibility}.
We then also show that among the feasible states the choice 
\begin{equation}
    \ket{\psi_q}_{A'B'} = \sqrt{q} \ket{0,0}_{A'B'} + \sqrt{\frac{1-q}{2}} \ket{1,1}_{A'B'} + \sqrt{\frac{1-q}{2}} \ket{2,2}_{A' B'} \, ,
\end{equation}
leads to the minimal entanglement entropy, that is, $H(A')_{\psi} = h(q)+1-q$, with $q=\frac{1}{1+\sin(\phi/2)}$.

Let $A$ and $B$ denote the two input registers. The protocol consists of the following four rounds of local operations followed by classical communication.\footnote{As we will see below, we need $4$ bits of classical communication. Alice and Bob send and receive both two bits of classical information. More precisely, Alices sends twice a single bit to Bob whereas Bob sends one ternary outcome (encoded in two bits) to Alice.} Its structure is shown in~\Cref{fig_new_teleportation}.

\vspace{2mm}
\noindent \textbf{Round 1 (Alice):}
Alice needs an ancilla qubit $\ket{0}_{\bar A}$, on which she later performs a measurement in the computational basis. 
Alice applies the unitary
\begin{align} \label{eq_first_unitary}
U_{\phi,1} = \ket{0}\bra{0}_A \otimes W^{(0)}_{A' \bar A} + \ket{1}\bra{1}_A \otimes W^{(1)}_{A' \bar A}
\end{align}
to $\bar A A A'$. The unitaries $W^{(0)}_{A' \bar A}$ and $W^{(1)}_{A' \bar A}$ are given in the basis $\cB=\{\ket{00},\ket{01},\ket{10},\ket{11},\ket{20},\ket{21}\}$ as
\begin{align}
W^{(0)}_{A' \bar A} = \begin{pmatrix}
\gamma_1 & -\gamma_2 & 0 & 0 & 0 & 0 \\
\gamma_2 & \gamma_1 & 0 & 0 & 0 & 0 \\
0 & 0 & 1 & 0 & 0 & 0 \\
0 & 0 & 0 & 0 & 1 & 0 \\
0 & 0 & 0 & 1 & 0 & 0 \\
0 & 0 & 0 & 0 & 0 & 1 
\end{pmatrix} \qquad \textnormal{and} \qquad 
W^{(1)}_{A' \bar A} = \begin{pmatrix}
\gamma_2 & -\gamma_1 & 0 & 0 & 0 & 0 \\
\gamma_1 & \gamma_2 & 0 & 0 & 0 & 0 \\
0 & 0 & 0 & 0 & 1 & 0 \\
0 & 0 & 1 & 0 & 0 & 0 \\
0 & 0 & 0 & 1 & 0 & 0 \\
0 & 0 & 0 & 0 & 0 & 1 
\end{pmatrix} \, ,
\end{align}
where \smash{$\gamma_1:= \sqrt{\frac{1/2 - p_1}{p_0}}$}, \smash{$\gamma_2:= \sqrt{\frac{1/2 - p_2}{p_0}}$}. Since $p_0 + p_1 + p_2 =1$ we have $\gamma_1^2 + \gamma_2^2 =1$, so both matrices are unitary. 
After this unitary, Alice measures the register $\bar A$ in the computational basis and sends the outcome $x \in \{0,1\}$ to Bob.
Conditioned on $x$ the support of $A'$ is two-dimensional. The measurement outcome $x=0$ occurs with probability $1/2$, independently of the input on $A$, and therefore reveals no information about that input.

\vspace{2mm}
\noindent \textbf{Round 2 (Bob):}
For outcome $x$, define
\begin{align} \label{eq_index_Bob}
 (i_x,j_x)&=\begin{cases}(1,2),&x=0\\(2,1),&x=1 \end{cases}
\end{align}
and the following operations in the computational basis of $B'$
\begin{align}
    P_0 = \begin{pmatrix}
        1 & 0 & 0\\
        0& 1 & 0 \\
        0& 0& 1
    \end{pmatrix}, \quad P_1 = \begin{pmatrix}
        1 & 0 & 0\\
        0& 0 & 1 \\
        0& 1 & 0
    \end{pmatrix}\,,  \quad      F_3= \frac{1}{\sqrt{3}}\begin{pmatrix}
        1 & 1 & 1\\
        1& \omega & \omega^2 \\
        1& \omega^2& \omega
    \end{pmatrix}, \quad 
    D_x = \begin{pmatrix}
        1 & 0 & 0\\
        0& d_{i_x} & 0 \\
        0& 0 & d_{j_x} 
    \end{pmatrix} \, .
\end{align}
Here, define $\omega := \ee^{2 \pi \ci /3}$ and phases $d_{i_x}$ and $d_{j_x}$ as specified in \cref{app_phases}.
Bob then applies the unitary
\begin{align}
 U_{\phi,2}^{(x)}
 =\proj{0}_B\otimes F_3P_x+\proj{1}_B\otimes F_3D_xP_x \, ,
 \label{eq:U2}
\end{align}
then measures $B'$ in the computational basis $\{\ket{0}_{B'},\ket{1}_{B'},\ket{2}_{B'} \}$ and sends the outcome $y\in \{0,1,2\}$ to Alice.

\vspace{2mm}
\noindent \textbf{Round 3 (Alice):}
After receiving $y$, Alice performs the unitary
\begin{equation}
    U_{\phi,3}^{(x,y)} = \ketbra{0}{0}_A \otimes \id_{A'} + \ketbra{1}{1}_A \otimes V^{(x,y)}_{A'} \, ,
\end{equation}
where 
\begin{equation}
    V_{A'}^{(x,y)} = \begin{pmatrix}
\displaystyle
\sqrt{\frac{\frac{1}{2}-p_{i_x}}{\frac{1}{2}-p_{j_x}}}
\frac{d_{j_x}-\ee^{\ci\phi}}{d_{j_x}-1}
&
\displaystyle
\frac{\sqrt{\frac{1}{2}-p_{i_x}}}{\sqrt{p_{j_x}}\omega^{2y}}
\frac{\ee^{\ci\phi}-1}{d_{j_x}-1}
\\[12pt]
\displaystyle
\frac{\sqrt{p_{i_x}}\omega^y}{\sqrt{\frac{1}{2}-p_{j_x}}}
\frac{d_{j_x}-\ee^{\ci\phi}d_{i_x}}{d_{j_x}-1}
&
\displaystyle
\sqrt{\frac{p_{i_x}}{p_{j_x}}}\omega^{-y}
\frac{e^{i\phi}d_{i_x}-1}{d_{j_x}-1}
\end{pmatrix}.
\end{equation}
Afterwards, she measures $A'$ in the computational basis and sends the outcome $z \in \{0,1\}$ to Bob.

\vspace{2mm}
\noindent \textbf{Round 4 (Bob):}
If Bob receives $z=1$, he implements the unitary
\begin{equation}
    U_{\phi,4}^{(x,z)} = \diag\left(1,\overline{d_{i_x}} \,\right)_B \, ,
\end{equation}
otherwise, he does nothing.   

\subsection{Correctness} \label{ssec:correct}
In this section, we prove that the protocol described above indeed implements the unitary $U_{\phi}$. To do so, we consider an arbitrary input state\footnote{By linearity it suffices to consider an arbitrary pure input state on the $AB$ system.} 
\begin{equation}
    \ket{\Lambda}_{AB} = \lambda_{00} \ket{00}_{AB} + 
    \lambda_{01} \ket{01}_{AB} +
    \lambda_{10} \ket{10}_{AB} +
    \lambda_{11} \ket{11}_{AB} \, .
\end{equation}
For reference,~\cref{fig_correctness} labels the conditional state after each measurement (conditioned on the specific measurement outcome).
\begin{figure}[!htb]
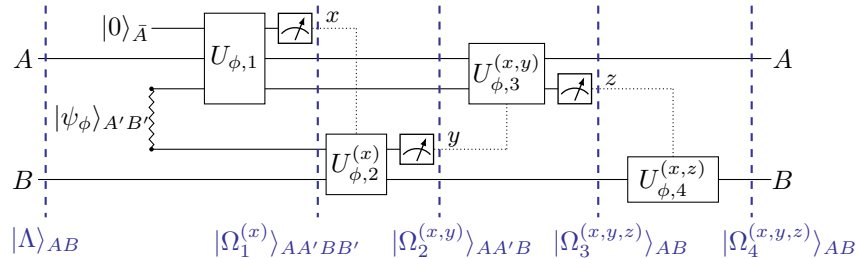

\centering
\include{Teleportation_correctness}
\vspace{-10mm}
\caption{Teleportation circuit for $U_{\phi}$. To prove correctness we need to show that $\ket{\Omega_4^{(x,y,z)}}_{AB}$ is proportional to $U_{\phi} \ket{\Lambda}_{AB}$ for all $x,z \in \{0,1\}$ and for all $y \in \{0,1,2\}$.}
\label{fig_correctness}
\end{figure}

\vspace{2mm}
\noindent \textbf{Round 1 (Alice):}
Given the measurement outcome $x \in \{0,1\}$, Alice effectively applies the Kraus operator
\begin{equation}
    \ketbra{0}{0}_A \otimes K^{(0)}_{A'} + \ketbra{1}{1}_A \otimes K^{(1)}_{A'} 
\end{equation}
with 
\begin{equation}
    K^{(0)}_{A'} = \sqrt{\frac{1/2-p_{i_x}}{p_0}} \ketbra{0}{0}_{A'} + \ketbra{1}{i_x}_{A'} \qquad \textnormal{and} \qquad K^{(1)}_{A'} =\sqrt{\frac{1/2-p_{j_x}}{p_0}} \ketbra{0}{0}_{A'} + \ketbra{1}{j_x}_{A'} 
\end{equation}
to the input state $\ket{\Lambda}_{AB} \otimes \ket{\psi_p}_{A'B'}$. 
Alice's operation, corresponding to outcome $x$, yields the unnormalized state
\begin{align} \label{eq_state_after_step1}
   \ket{\Omega_1^{(x)}}_{AA'BB'} = & \ket{0}_A (\lambda_{00} \ket{0}_B + \lambda_{01} \ket{1}_B) \otimes \left( \sqrt{1/2-p_{i_x}} \ket{0,0}_{A'B'}+ \sqrt{p_{i_x}} \ket{1,i_x}_{A'B'} \right) \nonumber \\
    &\hspace{10mm}+ \ket{1}_A (\lambda_{10} \ket{0}_B + \lambda_{11} \ket{1}_B) \otimes  \left(\sqrt{1/2-p_{j_x}} \ket{0,0}_{A'B'}+ \sqrt{p_{j_x}} \ket{1,j_x}_{A'B'}\right) .
\end{align}
Importantly, each measurement result occurs with probability $1/2$, independently of the input. Therefore, Alice does not learn any information about $A$.

\vspace{2mm}
\noindent \textbf{Round 2 (Bob):}
After Alice's measurement, the two A-branches (visible in~\cref{eq_state_after_step1}) lead to different qutrit coordinates for Bob's $B'$ system:
\begin{align}
    A=0 \rightarrow (0,i_x) \qquad \textnormal{and} \qquad A=1 \rightarrow (0,j_x)\, .
\end{align}
If Bob measured $B'$ at this point, he would obtain information about $A$ and destroy the required coherence. 
Therefore, he uses the Fourier basis, which treats all three qutrit coordinates equally.
Bob's first operation $P_x$ reorders the qutrit to $(0,i_x,j_x)$ and $D_x = \diag(1,d_{i_x},d_{j_x})$ implements the nonlocal phase of $U_\phi$, where the phases $d_{i_x}$ and $d_{j_x}$ are specified in \cref{app_phases}.
Bob then applies $F_3$ and measures $B'$.
After Bob obtains the measurement outcome $y \in \{0,1,2\}$, we get the unnormalized state
\begin{equation}
    \ket{\Omega_2^{{(x,y)}}}_{AA'B} = \lambda_{00} \ket{00}_{AB} \ket{l_0}_{A'} + \lambda_{01} \ket{01}_{AB} \ket{l_1}_{A'} + \lambda_{10} \ket{10}_{AB} \ket{r_0}_{A'} + \lambda_{11} \ket{11}_{AB} \ket{r_1}_{A'}
\end{equation}
with 
\begin{align}
    \ket{l_0} &= \sqrt{1/2 - p_{i_x}} \ket{0} + \sqrt{p_{i_x}} \omega^y \ket{1} \, , &&\ket{l_1} = \sqrt{1/2 - p_{i_x}} \ket{0} + \sqrt{p_{i_x}} \omega^y d_{i_x}\ket{1} \\
    \ket{r_0} &= \sqrt{1/2 - p_{j_x}} \ket{0} + \sqrt{p_{j_x}} \omega^{2y} \ket{1}\,, &&\ket{r_1} = \sqrt{1/2 - p_{j_x}} \ket{0} + \sqrt{p_{j_x}} \omega^{2y} d_{j_x} \ket{1} \, .
\end{align}
We suppress the common Fourier factor $1/\sqrt{3}$ for readability.

\vspace{2mm}
\noindent \textbf{Round 3 (Alice):}
%Alice performs a measurement conditioned on $A$.
Write the vectors $\mathbf{l}_0,\mathbf{l}_1,\mathbf{r}_0,\mathbf{r}_1 \in \C^{2\times 1}$ for the states defined above. For measurement outcome $z$ let $\mathbf{a}_z^T,\mathbf{b}_z^T \in \C^{1\times 2}$ denote the measurement rows used in the $A=0$ and $A=1$ branch.
The resulting state on $AB$ is, up to normalization, 
\begin{align} \label{eq_omega_3_early}
    \ket{\Omega_3^{(x,y,z)}}_{AB} = \lambda_{00} \mathbf{a}_z^T \mathbf{l}_0 \ket{00} +\lambda_{01} \mathbf{a}_z^T \mathbf{l}_1 \ket{01} + \lambda_{10} \mathbf{b}_z^T \mathbf{r}_0 \ket{10} + \lambda_{11}\mathbf{b}_z^T \mathbf{r}_1\ket{11} \, .
\end{align}
Bob's final correction will be a unitary of the form $U^{(x,y,z)}_{\phi,4} = \diag(u_0,u_1)_B$.
This applies the same correction $u_0$ to $\ket{00}$ and $\ket{10}$, so their amplitudes must already agree before the correction. Thus,
\begin{equation} \label{eq_cond0010}
    \mathbf{a}_z^T \mathbf{l}_0  = \mathbf{b}_z^T \mathbf{r}_0 \, .
\end{equation}
Similarly, $U^{(x,y,z)}_{\phi,4}$ applies $u_1$ to $\ket{01}$ and $\ket{11}$ and since $U_\phi$ applies $\ee^{\ci\phi}$ to $\ket{11}$, we also need
\begin{equation}\label{eq_cond0111}
    \ee^{\ci\phi} \mathbf{a}_z^T \mathbf{l}_1 =\mathbf{b}_z^T \mathbf{r}_1 
\end{equation}
before Bob's round. Define the matrices
\begin{equation}
    L := \begin{pmatrix}
        \mathbf{l}_0^T\\
        \ee^{\ci \phi} \mathbf{l}_1^T
    \end{pmatrix}\!=\! \begin{pmatrix}
        \sqrt{1/2 - p_{i_x}} & \sqrt{p_{i_x}} \omega^y \\
        \ee^{\ci\phi}\sqrt{1/2 - p_{i_x}} & \ee^{\ci \phi}\sqrt{p_{i_x}} \omega^y d_{i_x}
    \end{pmatrix}, \,   R :=  \begin{pmatrix}
        \mathbf{r}_0^T\\
        \mathbf{r}_1^T
    \end{pmatrix}\!=\! \begin{pmatrix}
        \sqrt{1/2 - p_{j_x}} & \sqrt{p_{j_x}} \omega^{2y} \\
        \sqrt{1/2 - p_{j_x}} & \sqrt{p_{j_x}} \omega^{2y} d_{j_x}
    \end{pmatrix},
\end{equation}
then \cref{eq_cond0010,eq_cond0111} can be written as $L \mathbf{a}_z = R \mathbf{b}_z$. If $R$ is invertible, we get $\mathbf{b}_z = R^{-1}L \mathbf{a}_z$.
Let us therefore define $T_{x,y} := R^{-1}L$.
Once Alice chooses her $A=0$ measurement basis, the phase conditions encoded in $T_{x,y}$ uniquely determine her $A=1$ basis.

To show that $L$ and $R$ are invertible, we calculate
\begin{align}
    \det L = \ee^{\ci \phi} \sqrt{1/2 - p_{i_x}}\sqrt{p_{i_x}} \omega^y(d_{i_x} - 1) \quad \textnormal{and} \quad
    \det R = \sqrt{1/2 - p_{j_x}}\sqrt{p_{j_x}} \omega^{2y}(d_{j_x} - 1) \, .
\end{align}
Since we assume $0<p_1,p_2 <1/2$ and $\phi \in (0,\pi]$,~\cref{eq_circle} implies $d_{i_x},d_{j_x} \ne 1$. Hence, both matrices are invertible.
We now derive~\cref{eq_circle} from requiring that $T_{x,y}$ be unitary -- otherwise Alice's operation would not be physical.

To satisfy $T_{x,y} T_{x,y}^\dagger= \id$, we need $LL^\dagger = RR^\dagger$. A simple calculation results in
\begin{equation}
    LL^\dagger = \begin{pmatrix}
        1/2 & \ee^{-\ci\phi}\big((1/2-p_{i_x})+p_{i_x} \overline{d_{i_x}} \,\big) \\
        \ee^{\ci\phi}\big((1/2-p_{i_x})+p_{i_x} d_{i_x}\big) & 1/2
    \end{pmatrix}
\end{equation}
and
\begin{equation}
    RR^\dagger = \begin{pmatrix}
        1/2 & (1/2-p_{j_x})+p_{j_x} \overline{d_{j_x}}\,\, \\
        (1/2-p_{j_x})+p_{j_x} d_{j_x} & 1/2
    \end{pmatrix} \, .
\end{equation}
Matching the non-diagonal terms then gives~\cref{eq_circle}.
Since $R$ is invertible, we can explicitly calculate
\begin{equation}
    T_{x,y} = \begin{pmatrix}
\displaystyle
\sqrt{\frac{\frac12-p_{i_x}}{\frac12-p_{j_x}}}
\frac{d_{j_x}-\ee^{\ci\phi}}{d_{j_x}-1}
&
\displaystyle
\frac{\sqrt{p_{i_x}}\omega^y}{\sqrt{\frac12-p_{j_x}}}
\frac{d_{j_x}-\ee^{\ci\phi}d_{i_x}}{d_{j_x}-1}
\\[12pt]
\displaystyle
\frac{\sqrt{\frac12-p_{i_x}}}{\sqrt{p_{j_x}}\omega^{2y}}
\frac{\ee^{\ci\phi}-1}{d_{j_x}-1}
&
\displaystyle
\sqrt{\frac{p_{i_x}}{p_{j_x}}}\omega^{-y}
\frac{\ee^{\ci\phi}d_{i_x}-1}{d_{j_x}-1}
\end{pmatrix}.
\end{equation}

The unitary $T_{x,y}$ fixes the ratios between the $\ket{00},\ket{01}$ amplitudes and between the $\ket{10},\ket{11}$ amplitudes. It remains to make their magnitudes equal. Alice therefore chooses the computational rows $\mathbf{a}_0^T=(1,0)$ and $\mathbf{a}_1^T=(0,1)$ for which
\begin{equation}\label{eq_same_magn}
    |\mathbf{a}_z^T \mathbf{l}_0| = |\mathbf{a}_z^T \mathbf{l}_1| \, ,
\end{equation}
because $\ket{l_0}$ and $\ket{l_1}$ differ only by a phase in their second coordinate. Setting $\mathbf{b}_z=T_{x,y} \mathbf{a}_z$, Alice's unitary is then given by
\begin{equation}
    U^{(x,y)}_{\phi,3} = \ketbra{0}{0}_A \otimes \id_{A'} + \ketbra{1}{1}_A \otimes (T_{x,y}^T)_{A'} \, .
\end{equation}

\vspace{2mm}
\noindent \textbf{Round 4 (Bob):}
Define $c_z = \mathbf{a}_z^T \mathbf{l}_0$ and $k_z = \mathbf{a}_z^T \mathbf{l}_1$ .
After Alice's final round,~\cref{eq_omega_3_early,eq_cond0010,eq_cond0111}, give
\begin{equation}
    \ket{\Omega_3^{(x,y,z)}}_{AB} = c_z \lambda_{00}\ket{00}_{AB} + k_z \lambda_{01} \ket{01}_{AB} + c_z \lambda_{10} \ket{10}_{AB} + \ee^{\ci\phi}k_z \lambda_{11}\ket{11}_{AB} \, .
\end{equation}
Therefore, Bob applies 
\begin{equation}
    U^{(x,y,z)}_{\phi,4} = \diag\left(\frac{\overline{c_z}}{|c_z|},\frac{\overline{k_z}}{|k_z|}\right)_B
\end{equation}
and since $|c_z|=|k_z|$ (see~\cref{eq_same_magn}),
this leads to the unnormalized state
\begin{equation}
    \ket{\Omega_4^{(x,y,z)}}_{AB} 
    = |c_z| \left( \lambda_{00}\ket{00}_{AB} +  \lambda_{01} \ket{01}_{AB} +  \lambda_{10} \ket{10}_{AB} + \ee^{\ci\phi}\lambda_{11}\ket{11}_{AB} \right)
    = |c_z| U_{\phi} \ket{\Lambda}_{AB} \, .
\end{equation}
Since this worked for all measurement outcomes $x,y,z$, the protocol implements $\diag(1,1,1,\ee^{\ci \phi})$ deterministically.
Note that $c_0=k_0=\sqrt{1/2-p_{i_x}}$ and $k_1=c_1 d_{i_x}=\sqrt{p_{i_x}} \omega^y d_{i_x}$.
Thus, Bob does nothing in the case of $z=0$ and applies \smash{$\diag(1,\overline{d_{i_x}}\,)$} otherwise.

%%%%%%%%%%%%%%%%%%%%%%%%%%%%%%%%%%%%%%%%%%%%%%%%%%%%%%%%%%%%%%%%%%%%%%%%%%%%%%%%%%%%
%%%%%%%%%%%%%%%%%%%%%%%%%%%%%%%%%%%%%%%%%%%%%%%%%%%%%%%%%%%%%%%%%%%%%%%%%%%%%%%%%%%%
\subsection{Minimal entanglement entropy}
Any state satisfying~\cref{eq_feasibility} can be used to implement $U_\phi$. In particular,
\begin{align}
   \ket{\psi_\phi}_{A'B'} = \sqrt{q} \ket{00} + \sqrt{\frac{1-q}{2}} \ket{11} + \sqrt{\frac{1-q}{2}} \ket{22}  \, ,
\end{align}
with \smash{$q=\frac{1}{1+\sin(\phi/2)}$}. This gives a resource state with $H(A')_\psi = h(q) + 1-q$ as discussed in~\cref{par:results}.
Note that no other choice of rank 3 state can perform better in our protocol. At fixed $p_0>1/2$, let $\Delta=p_1-p_2$.
The feasibility condition~\cref{eq_feasibility} then yields
\begin{align}\label{eq_entr_feas}
    \sin^2(\phi/2)(p_0^2-\Delta^2) \leq (1-p_0)^2 - \Delta^2 \, .
\end{align}   
At $p_0=q$,~\cref{eq_entr_feas} forces $\Delta=0$. Therefore, $\ket{\psi}_{A'B'}$ is the only rank 3 states that saturates the simple converse bound for $p_0$.

Fixing $p_0$ and setting $p_1 = \frac{1-p_0+\Delta}{2}$ and $p_2 = \frac{1-p_0-\Delta}{2}$ shows
\begin{align}
    \frac{\partial H(p)}{\partial \Delta} = \frac{1}{2} \log \frac{p_2}{p_1} \leq 0 \, .
\end{align}
Thus, at fixed $p_0$, making $p_1$ larger and $p_2$ smaller decreases the entropy. Therefore, a minimizing state cannot have slack in~\cref{eq_entr_feas}. If it did, one could increase $\Delta$ slightly and lower the entropy. Hence, any minimizer satisfies ~\cref{eq_entr_feas} with equality.
The equality case can be parametrized by $w = \sqrt{1-\cos^2(\phi/2)v^2}$ and $0\leq v \leq 1$ such that
\begin{align}
 (p_0,p_1,p_2)= p(v) = \frac{1}{\sin(\phi/2) + w} \left( w,\frac{\sin(\phi/2)(1+v)}{2}, \frac{\sin(\phi/2)(1-v)}{2} \right) \, ,
\end{align}
where the endpoints $v=0$ and $v=1$ correspond to $\ket{\psi_\phi}_{A'B'}$ and an ebit. Taking the second derivative of $H(p(v))$ while demanding $H'(p(v))=0$,  gives $H''(p(v))<0$.  The infimum is therefore at the boundary, whose entropy is $h(q)+1-q$ and $1$.
%%%%%%%%%%%%%%%%%%%%%%%%%%%%%%%%%%%%%%%%%%%%%%%%%%%%%%%%%%%%%%%%%%%%%%%%%%%%%%%%%%%
%%%%%%%%%%%%%%%%%%%%%%%%%%%%%%%%%%%%%%%%%%%%%%%%%%%%%%%%%%%%%%%%%%%%%%%%%%%%%%%%%%%
%%%%%%%%%%%%%%%%%%%%%%%%%%%%%%%%%%%%%%%%%%%%%%%%%%%%%%%%%%%%%%%%%%%%%%%%%%%%%%%%%%%
\section{Proof of the converse result} \label{sec_pf_converse}
Let the two input qubits of the gate be kept in the registers $A$ and $B$.  To
distinguish them from the resource registers, we write the pure resource state
in Schmidt form as \smash{$\ket{\psi_{\phi}}_{A'B'} =\sum_{j=0}^{d-1}\sqrt{p_j}\ket{j}_{A'}\ket{j}_{B'}$} with $p_0\geq\ldots \geq p_{d-1}\geq 0$ and \smash{$\sum_{j=0}^{d-1}p_j=1$}.
When both the gate input and the resource are present, the relevant bipartition is $AA':BB'$. Note that $H(A')_{\psi}=H(p)$.
Set $s:=\sin(\phi/2)$ and $q:=\frac{1}{1+s}$. Since $\phi\in[0,\pi]$, we have $s\in[0,1]$ and therefore
$q\in[1/2,1]$.  Consider the normalized two-qubit input state
\begin{align}
\ket{\eta_{\phi}}_{AB}
  :=\frac{1}{2}\big(
  \ket{00}+\ket{01}-\ee^{\ci\phi}\ket{10}+\ket{11}
  \big)\, .
\label{eq_converse_witness}
\end{align}
The Schmidt probabilities of $\ket{\eta_{\phi}}$ are $\lambda_{+}=\frac{1+s}{2}$ and $\lambda_{-}=\frac{1-s}{2}$.
Next, we apply the gate to this state.  Since
$U_{\phi}=\diag(1,1,1,\ee^{\ci\phi})$, we have
\begin{align}
\ket{\omega_{\phi}}_{AB}
  :=U_{\phi}\ket{\eta_{\phi}}_{AB}
  =\frac{1}{2}\big(
  \ket{00}+\ket{01}-\ee^{\ci\phi}\ket{10}
  +\ee^{\ci\phi}\ket{11}
  \big) \, ,
\label{eq_converse_witness_output}
\end{align}
which is maximally entangled and has Schmidt probability vector $(\frac{1}{2},\frac{1}{2})$.

By assumption, the teleportation protocol implements $U_{\phi}$ exactly and
deterministically on every input, including inputs that are entangled across
$A$ and $B$.  We may therefore run it on
$\ket{\eta_{\phi}}_{AB}$ while supplying the resource
$\ket{\psi}_{A'B'}$.  After the protocol, all local work registers and
classical records may be discarded locally.  Composing the protocol with
these local traces still gives an LOCC channel, and this channel maps the
initial density operator to $\proj{\omega_{\phi}}_{AB}$.  Equivalently, across
the bipartition $A A':B B'$, it performs the deterministic pure-state conversion
\begin{align}
\ket{\eta_{\phi}}_{AB}\otimes\ket{\psi}_{A'B'}
  \longmapsto
\ket{\omega_{\phi}}_{AB}\, .
\label{eq_converse_pure_state_conversion}
\end{align}
Nielsen's majorization criterion~\cite{Nielsen99} thus yields 
\begin{align}\label{eq_nielsen}
\ket{\eta_{\phi}}_{AB}\otimes\ket{\psi}_{A'B'} \preceq \ket{\omega_{\phi}}_{AB} \, .
\end{align}

The reduced state of a tensor-product pure state is the tensor product of the two reduced states.  Its eigenvalues, and hence its Schmidt probabilities, are therefore all pairwise products of the Schmidt probabilities of its two factors.  Hence, the Schmidt probabilities of the state on the left-hand side of~\cref{eq_nielsen} are $\{\lambda_{+}p_j,\lambda_{-}p_j:0\leq j\leq d-1\}$. Because $\lambda_{+}\geq\lambda_{-}$ and $p_0\geq p_j$ for every $j$, the largest one is $\lambda_{+}p_0=\frac{1+s}{2}p_0$.
The largest Schmidt probability of the output state is $\frac{1}{2}$.  
\Cref{eq_nielsen} therefore gives
\begin{align} \label{eq_converse_largest_schmidt_bound}
\frac{1+s}{2}p_0\leq\frac{1}{2} \quad \Longleftrightarrow \quad  p_0\leq\frac{1}{1+s}=q \, .
\end{align}

It remains to turn this into an entropy bound.  Pad probability vectors with zeros when necessary, and define $r:=(q,1-q,0,0,\ldots)$. This vector is decreasing because $q\geq1/2$.  We now directly check that the resource Schmidt vector $p=(p_0,p_1,\ldots)$ is majorized by $r$.  
For the first partial sum, \cref{eq_converse_largest_schmidt_bound} gives
\begin{align}
p_0\leq q=r_0.
\end{align}
For every partial sum containing at least two terms, we have for any $k \geq 2$
\begin{align}
\sum_{j=0}^{k-1}p_j
  \leq1
  =q+(1-q)
  =\sum_{j=0}^{k-1}r_j \, ,
\end{align}
and the total sums of both vectors are equal to one.  Hence $p$ is majorized
by $r$.  Since Shannon entropy is Schur concave~\cite[Section~3.D.1]{majorization_book}, majorization gives
\begin{align}
H(A')_{\psi}
=H(p)
\geq H(r)
=h(q)
=h\left(\frac{1}{1+\sin(\phi/2)}\right) \, .
\end{align}
\qed

%%%%%%%%%%%%%%%%%%%%%%%%%%%%%%%%%%%%%%%%%%%%%%%%%%%%%%%%%%%%%%%%%%%%%%%%%%%%%%%%%%
%%%%%%%%%%%%%%%%%%%%%%%%%%%%%%%%%%%%%%%%%%%%%%%%%%%%%%%%%%%%%%%%%%%%%%%%%%%%%%%%%%%
%%%%%%%%%%%%%%%%%%%%%%%%%%%%%%%%%%%%%%%%%%%%%%%%%%%%%%%%%%%%%%%%%%%%%%%%%%%%%%%%%%%
\section{Application: distributed quantum Fourier transform} \label{sec_QFT}
We illustrate the advantage of the novel teleportation protocol for controlled-phase gates with a distributed implementation of the QFT on $2n$ qubits, with Alice holding the first $n$ qubits and Bob the last $n$ qubits.  We write
\begin{align}
 \mathrm{QFT}_{2n}\ket{x}
 :=\frac{1}{2^n}\sum_{y=0}^{2^{2n}-1} \ee^{2\pi\ci xy/2^{2n}}\ket{y}\, ,
 \label{eq_qft_definition}
\end{align}
for $x\in\{0,\ldots,2^{2n}-1\}$.  Let $R_{2n}$ denote the reversal
permutation
\begin{align}
 R_{2n}\ket{y_1,\cdots, y_{2n}}
 :=\ket{y_{2n},\cdots, y_1}.
 \label{eq_qft_reversal}
\end{align}
The usual Hadamard--controlled-phase circuit implements this unitary up to a
final reversal of the order of the output qubits~\cite[Section~5]{nielsenChuang_book}.  We omit this reversal
and denote the resulting unitary by $\mathrm{QFT'}_{2n}$. 
Let 
\begin{align}
 e_{\mathrm{tel}}(\phi)
 :=\min\left\{1, h \Big(\frac{1}{1 + \sin(\phi/2)} \Big) + \frac{\sin(\phi/2)}{1+\sin(\phi/2)}\right\},
 \label{eq_qft_gate_cost}
\end{align}
denote the ebit consumption of the teleportation protocol for $U_{\phi}$ given in~\cref{fig_new_teleportation}. Thus, according to~\cref{eq_summary_results} we have $T_C(U_\phi)\leq e_{\mathrm{tel}}(\phi)$. Exact gate-teleportation protocols may be concatenated, and their pure resource states may be tensorized.  Since entanglement entropy is additive on tensor products, we have
\begin{align} \label{eq_subadditivity}
T_C(VU)\leq T_C(V)+T_C(U)
\end{align}
for any two unitaries $V$ and $U$.

In the standard QFT circuit~\cite[Section~5]{nielsenChuang_book}, every pair $j<k$ is acted upon by one controlled-phase gate $U_{\phi_d}$ for $d:=k-j$ and  $\phi_d:=\frac{2\pi}{2^{d+1}}=\frac{\pi}{2^d}$.
All Hadamard gates and all controlled-phase gates whose two qubits are held by
the same party are local and therefore free.  For a fixed $d$\footnote{As defined above, $d$ denotes the difference between the labels of two qubits.}, the
number of controlled-phase gates crossing the Alice--Bob cut is
\begin{align}
 \mu_n(d)
:=\left|\left\{(j,k):1\leq j\leq n<k\leq2n,\ k-j=d\right\}\right|
=\begin{cases}
d & 1 \leq d \leq n \\
2n-d & n <d \leq 2n -1 \, ,
\end{cases} 
 \label{eq_qft_multiplicity}
\end{align}
where $j$ and $k$ label one of Alice's and Bob's qubits, respectively. 
Indeed, $\mu_n(d)=d$ for $1\leq d\leq n$ and
$\mu_n(d)=2n-d$ for $n<d\leq2n-1$.  In particular,
\begin{align}
 \sum_{d=1}^{2n-1}\mu_n(d)=n^2,
 \label{eq_qft_number_cross_gates}
\end{align}
since each of Alice's $n$ qubits interacts once with each of Bob's $n$ qubits.

Applying the new teleportation protocol (see~\cref{fig_new_teleportation}) independently to these nonlocal gates
gives
\begin{align}
 T_C\left(\mathrm{QFT'}_{2n}\right)
 \overset{\textnormal{\Cshref{eq_subadditivity}}}{\leq} 
 \sum_{d=1}^{2n-1} \mu_n(d) e_{\mathrm{tel}}\left(\frac{\pi}{2^d}\right)
 =: E^{(n)}_{\mathrm{QFT'}_{2n}} \, .
 \label{eq_qft_exact_cost}
\end{align}
Clearly $n \mapsto E^{(n)}_{\mathrm{QFT'}_{2n}}$ is monotonically increasing in $n$. To bound the teleportation cost, recall that $\mu_n(d) \leq d$ and hence
\begin{align}
E^{(n)}_{\mathrm{QFT'}_{2n}}
\leq \lim_{n \to \infty} E^{(n)}_{\mathrm{QFT'}_{2n}}
\overset{\textnormal{\Cshref{eq_qft_exact_cost}}}{\leq} \sum_{d=1}^\infty d \,  e_{\mathrm{tel}}\left(\frac{\pi}{2^d}\right)
\overset{\textnormal{\Cshref{eq_qft_gate_cost}}}{\approx} 12.1869 \, \textnormal{ebits} \, .
\end{align}
For comparison, the construction that applies the previous one-ebit protocol from~\cite[Theorem~2]{Eisert00} separately to all nonlocal controlled unitaries uses $n^2$ ebits. 
If we allow for an error $\eps>0$, the approximate QFT~\cite{QFT94} can be implemented using the one-ebit protocol from~\cite[Theorem~2]{Eisert00} at a cost of $O(\log^2(n/\eps))$ ebits. 

%%%%%%%%%%%%%%%%%%%%%%%%%%%%%%%%%%%%%%%%%%%%%%%%%%%%%%%%%%%%%%%%%%%%%%%%%%%%%%%%%%%
%%%%%%%%%%%%%%%%%%%%%%%%%%%%%%%%%%%%%%%%%%%%%%%%%%%%%%%%%%%%%%%%%%%%%%%%%%%%%%%%%%%
\paragraph{Acknowledgments} We thank Shao-Hua Hu, Christophe Piveteau, and Jun-Yi Wu for useful discussions on gate teleportation. Furthermore, we acknowledge GPT 5.6 Sol for assistance in the development of our protocol and the tighter version of the converse (presented in~\cref{app_improved_converse}). L.S. acknowledges support from the Quantum Center at ETH Zurich.
%%%%%%%%%%%%%%%%%%%%%%%%%%%%%%%%%%%%%%%%%%%%%%%%%%%%%%%%%%%%%%%%%%%%%%%%%%%%%%%%%%%
%%%%%%%%%%%%%%%%%%%%%%%%%%%%%%%%%%%%%%%%%%%%%%%%%%%%%%%%%%%%%%%%%%%%%%%%%%%%%%%%%%%

\appendix
\section{How to choose the phases in round 2}
\label{app_phases}
The phases $d_{i_x}$ and $d_{j_x}$ in round $2$ of the teleportation protocol are solutions to
\begin{align}\label{eq_circle}
    \ee^{\ci\phi}\Big(\frac{1}{2}-p_{i_x}+p_{i_x} d_{i_x}\Big) = \frac{1}{2} - p_{j_x}  + p_{j_x} d_{j_x} \, .
\end{align}
This is an equation about circles intersecting in the complex plane and it arises from a unitarity condition in Alice's round 3 (see \cref{ssec:correct}). Note that~\cref{eq_circle} has a solution when the circles intersect. This is the case when the distance between their centers
$
    \delta = | \ee^{\ci \phi}\left(1/2-p_{i_x}\right) - \left(1/2-p_{j_x}\right)|
$
satisfies
\[
    |p_{i_x} - p_{j_x}| \leq \delta \leq p_{i_x} + p_{j_x} .
\]
Since 
\begin{equation}
    \delta^2 = (p_{i_x} - p_{j_x})^2 + 4\Big(\frac{1}{2}-p_{i_x}\Big) \Big(\frac{1}{2}-p_{j_x}\Big)\sin^2(\phi/2) \, ,
\end{equation}
the lower inequality follows automatically, and the upper inequality gives the feasibility condition presented in~\cref{eq_feasibility}.
\begin{equation}
    \sin^2(\phi/2)(1-2p_1)(1-2p_2)\leq 4p_1p_2 \, .
\end{equation}
Any rank-3 resource state that satisfies this condition can be used with our protocol to implement $U_\phi$. 
If the feasibility condition is satisfied, let us define
\begin{align}
    u := \frac{ (1/2-p_{j_x}) -  \ee^{\ci \phi}(1/2-p_{i_x}) }{\delta}\, , \quad  
    \ell := \frac{p_{i_x}^2- p_{j_x}^2+\delta^2}{2\delta}\,  \quad \textnormal{and} \quad
    h := \sqrt{p_{i_x}^2-\ell^2} \, .
\end{align}
An intersection point is then given by $v = \ee^{\ci\phi}(\frac{1}{2}-p_{i_x}) + (\ell + \ci h)u$ and we find
\begin{align}
    d_{i_x} = \frac{\ee^{-\ci \phi}v - (1/2-p_{i_x})}{p_{i_x}} \quad \textnormal{and} \quad d_{j_x} = \frac{v - (1/2-p_{j_x})}{p_{j_x}} \, .
\end{align}

\section{A tighter converse} \label{app_improved_converse}
In~\cref{sec_pf_converse}, we have seen that the largest Schmidt probability of a resource state has to fulfill $p_0\leq q := 1/(1+\sin(\phi/2))$. 
It is possible to tighten this converse if we constrain the resource state to have a fixed Schmidt rank $D \geq 3$.
In particular we will see that for every resource state with finite Schmidt rank, there are angles $\phi$ such that $U_{\phi}$ requires at least one ebit.
\Cref{fig_results_curve_stronger_converse} plots the improved converse derived in~\cref{eq_stronger_converse}.
\begin{figure}[!htb]
\centering
\input{Fig_results_stronger_converse}
\caption{Entanglement of the resource state for teleporting a controlled-phase gate $U_{\phi}=\diag(1,1,1,\ee^{\ci \phi})$ of angle $\phi \in [0,\pi]$. The achievability curve corresponds to the protocol in~\cref{fig_new_teleportation}. The converse curves show~\cref{eq_stronger_converse} for resource states of Schmidt rank $D\in\{3,4,10,\infty\}$. For such a state, gate teleportation below the corresponding curve is impossible.}
\label{fig_results_curve_stronger_converse}
\end{figure}
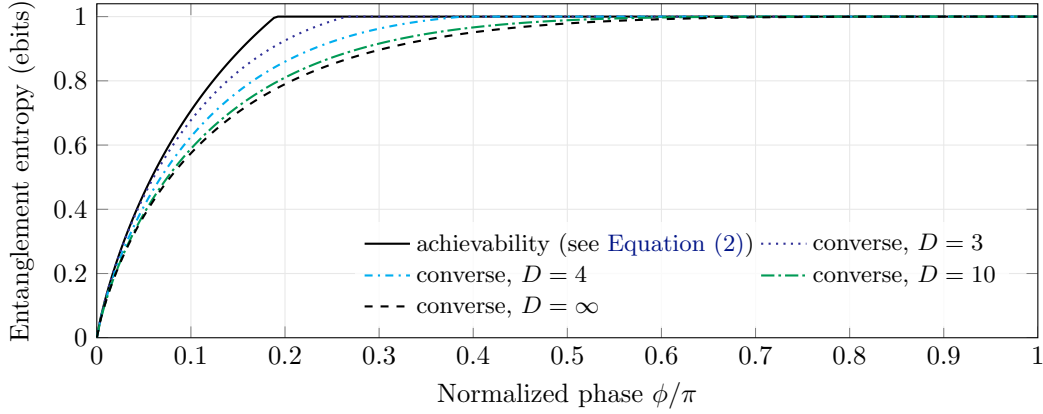

To prove the stronger converse bound, we start with the following simple facts about entropy.
\begin{lemma}\label{lem_tailbound}
Let $p = (p_0,p_1,\dots,p_{n-1}) \in \R_+^n$ be a probability distribution with $p_0 \geq p_1 \geq \ldots p_{n-1}$, $r = 1- p_0$, and $w_i = p_i/r$ for all $i \in \{1,\ldots,n-1\}$. Then
\begin{enumerate}[(i)]
\item $ H(p) = h(p_0) + rH(w) $ and for $\tau = \sum_{i=2}^{n-1} p_i$ we have $H(p) \geq h(p_0) + rh(\tau/r)$ \label{it_entropy1}
\item if $p_i \leq \frac{1}{2}$ for all $i$ we have $H(p) \geq 1$. \label{it_entropy2}
\end{enumerate}
\end{lemma}
\begin{proof}
We start by proving the statement~\eqref{it_entropy1}.
Note that $1-p_0 = \sum_{i=1}^{n-1} p_i$ and therefore $w=(w_1,\ldots,w_{n-1}) \in \R_+^{n-1}$ is a probability distribution. Thus, we have
\begin{align}
H(p)
&= - p_0 \log p_0 - \sum_{i=1}^{n-1} p_i \log p_i \\
&=- p_0 \log p_0 - \sum_{i=1}^{n-1} (r w_i) \log (r w_i) \\
&=- p_0 \log p_0 - r\log r\sum_{i=1}^{n-1} w_i -r\sum_{i=1}^{n-1} w_i  \log w_i \\
&= h(p_0) + r H(w) \, .
\end{align}
The second statement of~\eqref{it_entropy1} follows from the first by grouping $w_2,w_3,\dots,w_{n-1}$ into one event of probability $\tau/r$. Then define the probabilities $v_i = p_i/\tau$ for $i\geq 2$. A second application of the previous calculation gives
\begin{equation}
    H(w) = h(\tau/r) + \frac{\tau}{r} H(v) \geq h(\tau/r) \, .
\end{equation}

Thus, it remains to prove~\eqref{it_entropy2}. Note that $p_i \leq \frac{1}{2}$ implies $-\log p_i \geq 1$. Therefore
\begin{align}
H(p) 
= -\sum_{i=0}^{n-1}p_i \log p_i
\geq \sum_{i=0}^{n-1} p_i
=1 \, .
\end{align}
\end{proof}

For $r=1-p_0$ let $A := \sin(\phi/2)\frac{2p_0-1}{2}$ and $B:= A(D-3)+r(D-2)$. Furthermore, let
    \begin{equation} \label{eq_tau_opt}
\tau_D(p_0) := 2\, \frac{B - \sqrt{B^2-(D-1)^2A^2}}{(D-1)^2} \, .
    \end{equation}
The heart of the converse bound is built on the following lemma. The first steps are similar to~\cite{SG11}.
\begin{lemma}\label{lem_conv_core}
    The resource state of every deterministic SEP implementation, hence every LOCC teleportation, satisfies
    \begin{equation} \label{eq_strong_converse_feasible}
        \sin(\phi/2)(2p_0-1) \leq 2\sum_{1\leq i < j}^{D-1} \sqrt{p_i p_j} \, .
    \end{equation}
\end{lemma}
\begin{proof}
    Every LOCC protocol is included in the class of SEP channels~\cite{LOCC14}, which we can characterize via product Kraus operators $K_k = E_k \otimes F_k$ satisfying
\begin{equation} \label{eq_kraus1}
    \sum_k K_k^\dagger K_k = \id \, .
\end{equation}
Here, $E_k : AA' \mapsto A$ and $F_k :BB' \mapsto B$.
For a deterministic implementation of $U_\phi$, each of those operators has to satisfy
\begin{equation} \label{eq_kraus2}
    K_k ( \id_{AB} \otimes \ket{\psi_p}_{A'B'}) = \alpha_k U_\phi\,,
\end{equation}
where 
\begin{equation}
    \ket{\psi_p}_{A'B'} = \sum_{i=0}^{D-1} \sqrt{p_i}\ket{ii}_{A'B'}
\end{equation}
is the resource state.
The scalar $\alpha_k$ can depend on $k$, but not on the input. Since $U_\phi$ is unitary \cref{eq_kraus1,eq_kraus2} imply
\begin{equation}\label{eq_kraus3}
    \sum_k |\alpha_k|^2 = 1 \, .
\end{equation}
In the following, we will absorb the phase into each Kraus operator so that $\alpha_k \geq 0$.
Next, let us express $U_\phi$ as
\begin{equation}
    U_\phi = \sum_{a,b=0}^1 \hat{U}_{ab} \ketbra{a}{a}_A \otimes \ketbra{b}{b}_B \qquad \textnormal{with} \qquad \hat{U} = \begin{pmatrix}
        1 & 1 \\
        1 & \ee^{\ci\phi}
    \end{pmatrix} \, .
\end{equation}
Note that, $\hat{U}$ is invertible since $\det \hat{U} = \ee^{\ci \phi}-1 \neq 0$.\footnote{Without loss of generality we assume $\phi \in (0,\pi]$, since the case $\phi=0$ is trivial anyway.}
For each $k$ we use the basis of $\ket{a}_A$ and $\ket{b}_B$ to sandwich~\cref{eq_kraus2} yielding
\begin{align}
     \bra{a}_A\bra{b}_B (E_k\otimes F_k) ( \ket{a}_A\ket{b}_B \otimes \ket{\psi_p}_{A'B'}) = \alpha_k \hat{U}_{ab} \, .
\end{align}
Writing $\ket{\psi_p}_{A'B'}$ in its Schmidt decomposition, we get
\begin{equation}\label{eq_consist1}
    \sum_{i=0}^{D-1} \sqrt{p_i} \, \bra{a}_A E_k \ket{a}_A \ket{i}_{A'} \, \bra{b}_B F_k \ket{b}_B \ket{i}_{B'} = \alpha_k \hat{U}_{ab} \, .
\end{equation}
Defining coefficients 
\begin{equation}
    (e_k)_{ai}:=  \bra{a}_A E_k \ket{a}_A \ket{i}_{A'}  \qquad \textnormal{and} \qquad 
    (f_k)_{bi}:= \bra{b}_B F_k \ket{b}_B \ket{i}_{B'}
\end{equation}
we can write~\cref{eq_consist1} as
\begin{equation}\label{eq_consist2}
    \sum_{i=0}^{D-1} \sqrt{p_i} \, (e_k)_{ai} \, (f_k)_{bi} = \alpha_k \hat{U}_{ab} \, ,
\end{equation}
or if we define the diagonal matrix $\Psi = \diag(\sqrt{p_0},\sqrt{p_1},\sqrt{p_2},\dots)$ as
\begin{equation}\label{eq_matrix}
    \hat{E}_k \Psi \hat{F}_k^T = \alpha_k \hat{U} \, .
\end{equation}
Here, $\hat{E}_k, \hat{F}_k \in \mathbb{C}^{2\times D}$ and we can write them as matrices whose columns are two-dimensional vectors
\begin{equation}
    \hat{E}_k = 
    \begin{pmatrix}| & &| \\
        \mathbf{e}_{k,0} & \cdots & \mathbf{e}_{k,D-1} \\
        | & &|
    \end{pmatrix} \qquad \textnormal{and} \qquad  
    \hat{F}_k = 
    \begin{pmatrix}| & &| \\
        \mathbf{f}_{k,0} & \cdots & \mathbf{f}_{k,D-1} \\
        | & &|
    \end{pmatrix} \, .
\end{equation}
With this, we can write~\cref{eq_matrix} as
\begin{equation}\label{eq_consist1_vec}
    \sum_{i=0}^{D-1} \sqrt{p_i} \, \mathbf{e}_{k,i}\mathbf{f}_{k,i}^T = \alpha_k \hat{U} \, .
\end{equation}

Let us now take the determinant of both sides of~\cref{eq_matrix}. This gives
\begin{equation}
    \det (\alpha_k \hat{U}) = \alpha_k^2 \det \hat{U} 
\end{equation}
and on the other side using Binet-Cauchy~\cite[Section~0.8.7]{Horn_Johnson_1985}
\begin{equation}
    \det ( \hat{E}_k \Psi \hat{F}_k^T) = \sum_{i<j} \sqrt{p_i p_j}  \det(\mathbf{e}_{k,i}\, \mathbf{e}_{k,j}) \det(\mathbf{f}_{k,i}\, \mathbf{f}_{k,j}) \, .
\end{equation}
For sake of notation, define 
\begin{equation}
    \Lambda_{ij}^k :=\det(\mathbf{e}_{k,i}\, \mathbf{e}_{k,j}) \det(\mathbf{f}_{k,i}\, \mathbf{f}_{k,j}) 
\end{equation}
If we now sum over $k$, we get
\begin{equation}\label{eq_conv_det1}
    \sum_k \alpha_k^2 \det \hat{U}  \overset{\textnormal{\Cshref{eq_kraus3}}}{=}  \det \hat{U} = \sum_{i<j} \sqrt{p_i p_j} \, \sum_k \Lambda_{ij}^k \, .
\end{equation}

Starting from~\cref{eq_consist1_vec}, we can also define the quantity $M(t)$
\begin{equation}
M(t) 
:= t  \sqrt{p_0} \, \mathbf{e}_{k,0}\mathbf{f}_{k,0}^T +\sum_{i=1}^{D-1} \sqrt{p_i} \mathbf{e}_{k,i}\mathbf{f}_{k,i}^T \, ,
\end{equation}
which for $t=1$ again gives~\cref{eq_consist1_vec}. Calculating the determinant of this expression then gives analogously as above
\begin{equation}
    \det M(t) = t \sum_{j=1}^{D-1}
 \sqrt{p_0 p_j}\, \Lambda^k_{0j} + \sum_{1\leq i < j \leq D-1} \sqrt{p_i p_j}\,  \Lambda^k_{ij} \, .
\end{equation}
For $2\times 2$ matrices, we have
\begin{equation}
    \frac{\di }{\di t} \det M(t) = \tr[\adj(M(t))M'(t)]
\end{equation}
where $\adj(M)$ denotes the adjugate of $M$.
Evaluating this expression at $t=1$ with $M'(t)= \sqrt{p_0} \, \mathbf{e}_{k,0}\mathbf{f}_{k,0}^T $
gives
\begin{equation}\label{eq_t0_equation}
    \sum_{j=1}^{D-1} \sqrt{p_0 p_j}\, \Lambda^k_{0j} =\sqrt{p_0} \, \tr[\adj(\alpha_k \hat{U})  \mathbf{e}_{k,0}\mathbf{f}_{k,0}^T]
\end{equation}
For $2\times 2$ matrices, the adjugate satisfies $\adj(\alpha_k \hat{U}) = \alpha_k  \adj( \hat{U}) = \alpha_k \det(\hat{U}) \hat{U}^{-1}$ and, therefore, we get
\begin{equation}\label{eq_conv_det}
     \sum_{j=1}^{D-1}\sqrt{p_0 p_j}\, \Lambda^k_{0j} = \sqrt{p_0} \det (\hat{U}) \tr[\hat{U}^{-1} \alpha_k\mathbf{e}_{k,0}\mathbf{f}_{k,0}^T]\, .
\end{equation}
To proceed, we return to~\cref{eq_kraus1}.
Multiply it on the left by $\id_{AB} \otimes \bra{\psi_p}_{A'B'}$ to get
\begin{align}
\id_{AB} \otimes \bra{\psi_p}_{A'B'} 
\overset{\textnormal{\Cshref{eq_kraus1}}}{=} \sum_k \id_{AB} \otimes \bra{\psi_p}_{A'B'} K_k^\dagger K_k 
\overset{\textnormal{\Cshref{eq_consist1}}}{=} \sum_k \alpha_k U^\dagger_\phi K_k  \, .
\end{align}
Another multiplication with $U_\phi$ gives
\begin{equation}
     \sum_k \alpha_k  K_k  = (U_\phi)_{AB} \otimes \bra{\psi_p}_{A'B'} \, .
\end{equation}
Now we sandwich this equation with basis vectors $\ket{a}_A,\ket{b}_B,\ket{ij}_{A'B'}$ again to obtain
\begin{align}
    \sum_k \alpha_k \, \bra{a}_A\bra{b}_B \, K_k \, ( \ket{a}_A\ket{b}_B \otimes \ket{ij}_{A'B'}) &= \bra{a}_A\bra{b}_B(U_\phi)_{AB} \ket{a}_A\ket{b}_B  \,\braket{\psi_p|ij}_{A'B'} \, ,
\end{align}  
which equivalently can be expressed as
\begin{align}
    \sum_k \alpha_k (e_k)_{ai} (f_k)_{bj} &= \hat{U}_{ab} \sqrt{p_i} \, \delta_{ij}\, .\label{eq_completeness}
\end{align}
Choosing $i=j=0$, this gives
\begin{equation}
    \sum_k \alpha_k (e_k)_{a0} (f_k)_{b0} = \sqrt{p_0} \hat{U}_{ab} 
\end{equation}
or as matrices with the vector notation from before
\begin{equation}\label{eq_converse3}
    \sum_k \alpha_k  \,\mathbf{e}_{k,0} \mathbf{f}_{k,0}^T = \sqrt{p_0} \hat{U} \, .
\end{equation}
Going back to~\cref{eq_conv_det}, we note that summing over $k$ gives
\begin{align}
    \sum_{j=1}^{D-1}\sqrt{p_0 p_j}\, \sum_k\Lambda^k_{0j} &=  \sqrt{p_0} \det \hat{U} \, \tr[\hat{U}^{-1}\sum_k\alpha_k\mathbf{e}_{k,0} \mathbf{f}_{k,0}^T] \\
    \overset{\textnormal{\Cshref{eq_converse3}}}&{=} \sqrt{p_0} \det \hat{U}  \,\tr[\hat{U}^{-1} \, \sqrt{p_0} \hat{U} ]  \\
    &= p_0 \det \hat U \tr[\id_2] \\
    &= 2 p_0 \det \hat{U}\, .
\end{align}
Subtracting this term from~\cref{eq_conv_det1} then gives
\begin{equation}\label{eq_conv_det3}
    \sum_{1\leq i< j}^{D-1} \sqrt{p_ip_j}\sum_k \Lambda_{ij}^k = (1-2p_0) \det \hat U \, .
\end{equation}
To finish our proof, we need to bound $\sum_k \Lambda_{ij}^k $.
To do so, we expand it into scalar terms
\begin{align}\label{eq_conv_lambda}
     \Lambda_{ij}^k &= \det(\mathbf{e}_{k,i}\, \mathbf{e}_{k,j}) \det(\mathbf{f}_{k,i}\, \mathbf{f}_{k,j}) \\
     &= ((e_k)_{0i}  \,(f_k)_{0i})  ((e_k)_{1j} \, (f_k)_{1j}) -  ((e_k)_{0i}  \,(f_k)_{1i})  ((e_k)_{1j} \, (f_k)_{0j}) \\
     & \qquad -((e_k)_{1i}  \,(f_k)_{0i})  ((e_k)_{0j} \, (f_k)_{1j}) + ((e_k)_{1i}  \,(f_k)_{1i})  ((e_k)_{0j} \, (f_k)_{0j}) \, .
\end{align}
Furthermore, we note that~\cref{eq_kraus1} implies
\begin{equation}
    \sum_k \norm{E_k \ket{ai}_{AA'}}^2 \norm{F_k \ket{bj}_{BB'}}^2 = 1 
\end{equation}
by multiplying $\ket{ai}_{AA'}\ket{bj}_{BB'}$ from both sides.
Projecting $E_k \ket{ai}_{AA'}$ onto $\ket{a}_A$ and respectively with  $F_k \ket{bj}_{BB'}$ onto $\ket{b}_B$ can only decrease their norms, so this implies
\begin{equation}\label{eq_conv_sumbound}
    \sum_k |(e_k)_{ai}|^2 |(f_k)_{bj}|^2 \leq 1 \, .
\end{equation}
Using Cauchy-Schwarz and~\cref{eq_conv_sumbound}, we can bound the first term in~\cref{eq_conv_lambda} by
\begin{equation}
    \sum_k |((e_k)_{0i}  \,(f_k)_{0i})  ((e_k)_{1j} \, (f_k)_{1j})|
    \leq \sqrt{\sum_k |(e_k)_{0i}|^2 |(f_k)_{0i}|^2} \sqrt{\sum_k |(e_k)_{1j}|^2 |(f_k)_{1j}|^2} \leq 1 \, .
\end{equation}
The same argument also applies to the other three terms. Therefore,
\begin{equation} \label{eq_Lambda_bound}
    \sum_k |\Lambda^k_{ij}| \leq 4 .
\end{equation}
Now using the triangle inequality we get from~\cref{eq_conv_det3}
\begin{equation}
    (2p_0-1) \, |\det \hat U|  \leq \sum_{1\leq i<j}^{D-1} \sqrt{p_i p_j} \sum_k |\Lambda^k_{ij}| \overset{\textnormal{\Cshref{eq_Lambda_bound}}}{\leq} 4\sum_{1\leq i<j}^{D-1}  \sqrt{p_i p_j} \, .
\end{equation}
Together with $|\det \hat{U}| = 2 \sin(\phi/2)$, this completes the proof.
\end{proof}

\begin{lemma}\label{lem_cauchy} Let $p = (p_0,p_1,\dots,p_{D-1}) \in \R_+^D$ be a probability distribution with $p_0 \geq p_1 \geq \ldots p_{D-1}$, $r = 1- p_0$, and $\tau = \sum_{i=2}^{D-1}p_i$. Then
    \begin{equation}
        \sum_{1\leq i< j}^{D-1} \sqrt{p_i p_j}  \leq  \sqrt{(r-\tau)(D-2)\tau} + \frac{D-3}{2}\tau =: f(\tau) \, .
    \end{equation}
    Moreover, $0 \leq \tau \leq \frac{(D-2)r}{D-1}$ and $f$ is increasing on this interval.
\end{lemma}
\begin{proof}
    Separating the $p_1$ term, we get
    \begin{equation} \label{eq_lem_cauchy1}
         \sum_{1\leq i< j}^{D-1} \sqrt{p_i p_j} =  \sqrt{p_1} \sum_{i=2}^{D-1} \sqrt{p_i} +  \sum_{2\leq i< j}^{D-1} \sqrt{p_i p_j} \, .
    \end{equation}
    Cauchy-Schwarz then gives
    \begin{equation} \label{eq_up_tau_1}
        \sum_{i=2}^{D-1} \sqrt{p_i} \leq \sqrt{(D-2)\sum_{i=2}^{D-1} p_i } = \sqrt{(D-2)\tau} \, .
    \end{equation}
    Consider now 
    \begin{equation}
        \left( \sum_{i=2}^{D-1} \sqrt{p_i} \right)^2 = \sum_{i=2}^{D-1} p_i + 2 \sum_{2\leq i < j} \sqrt{p_i p_j} = \tau + 2 \sum_{2\leq i < j} \sqrt{p_i p_j} \, .
    \end{equation}
    Using Cauchy-Schwarz on the first term and rearranging gives us
    \begin{equation}
        \sum_{2\leq i < j} \sqrt{p_i p_j} \leq  \frac{D-3}{2} \tau\, . \label{eq_up_tau_2}
    \end{equation}
    Together, noting that $p_1 = r-\tau$, we can rewrite~\cref{eq_lem_cauchy1} as
    \begin{equation}
         \sum_{1\leq i< j}^{D-1} \sqrt{p_i p_j}  
         \overset{\textnormal{\Cshref{eq_lem_cauchy1,eq_up_tau_1,eq_up_tau_2}}}{\leq}  \sqrt{(r-\tau)(D-2)\tau} + \frac{D-3}{2}\tau \, .
    \end{equation}

Since $p_i \leq p_1 = r-\tau $ for all $i\geq2$, this implies $\tau \leq (D-2)(r-\tau)$, which rearranges to
$\tau \leq \frac{(D-2)r}{D-1}$. Taking the derivative of the defined function $f$, one can check that $f'(\tau)\geq0$ on the interval $\tau\in[0,\frac{(D-2)r}{D-1}]$.
\end{proof}

\begin{lemma}\label{lem_taubound}
Let $p \in \R_+^D$ be a resource spectrum satisfying~\cref{eq_strong_converse_feasible} with $1/2 \leq p_0 \leq q$, $r = 1- p_0$, and $\tau = \sum_{i=2}^{D-1}p_i$. Then,
    \begin{equation} \label{eq_bound_tau}
        \tau \geq \tau_D(p_0)\, , 
    \end{equation}
    where $\tau_D(p_0)$ is defined in~\cref{eq_tau_opt}.
\end{lemma}
\begin{proof}
We start by noting that 
    \begin{equation}
        \sin(\phi/2)\frac{2p_0-1}{2} 
        \overset{\textnormal{\Cshref{lem_conv_core}}}{\leq} \sum_{1\leq i < j}^{D-1} \sqrt{p_i p_j}  
        \overset{\textnormal{\Cshref{lem_cauchy}}}{\leq} f(\tau) \, .
    \end{equation}
    By~\cref{lem_cauchy}, $f$ is increasing on $[0,\frac{(D-2)r}{D-1}]$ and takes values in $[f(0)=0,f(\frac{D-2}{D-1}r)=\frac{D-2}{2}r]$. So, there is one $\tau_*$ such that
    \begin{equation} \label{eq_def_A}
        f(\tau_*) = \sin(\phi/2) \frac{2p_0-1}{2} = A
    \end{equation}
    or equivalently
    \begin{equation}
        \sqrt{(r-\tau_*)(D-2)\tau_*} = A - \frac{D-3}{2} \tau_* \, .
    \end{equation}
This can be reformulated as the following quadratic equation
    \begin{align} \frac{(D-1)^2}{4}\tau_*^2 - \big(A(D-3)+r(D-2) \big)\tau_* + A^2=0 \, .
    \end{align}
    Solving it for the smaller root (it is the one that vanishes as $A\to0$) yields
    \begin{equation}
        \tau_* 
        = 2\, \frac{B - \sqrt{B^2-(D-1)^2A^2}}{(D-1)^2} 
        \overset{\textnormal{\Cshref{eq_tau_opt}}}{=} \tau_D(p_0) \, ,
    \end{equation}
    where $B = A(D-3)+r(D-2)$. Together, we have
    \begin{equation}
        f(\tau_*) \overset{\textnormal{\Cshref{eq_def_A}}}{=} A \leq f(\tau)
    \end{equation}
    and monotonicity of $f$ implies $\tau\geq \tau_*$.
\end{proof}

\begin{theorem}\label{thm_tightestconverse}
    For $1/2\leq p_0 \leq q<1$, let $\tau_D(p_0)$ be as defined in~\cref{eq_tau_opt}. Every rank-D spectrum satisfying~\cref{eq_strong_converse_feasible} fulfills
    \begin{align} \label{eq_converse_stronger}
        H(p) \geq h(p_0) + (1-p_0)h\left(\frac{\tau_D(p_0)}{1-p_0}\right) \, .
    \end{align}
\end{theorem}
\begin{proof}
    Let $r=1-p_0$ and $\tau$ be the mass outside the two largest probabilities. By~\cref{lem_taubound}, $\tau \geq \tau_D(p_0)$. 
    First, suppose $\tau \leq r/2$. Then both $\tau_D/r$ and $\tau/r$ are in $[0,1/2]$. The binary entropy increases on this interval.
    Using \cref{lem_tailbound},
    \begin{equation}
        H(p)\geq h(p_0) + rh\left(\frac{\tau}{r}\right) \geq h(p_0) + r h\left(\frac{\tau_D(p_0)}{r}\right) \, .
    \end{equation}

    Now suppose $\tau > r/2$. Then $p_1=r-\tau <r/2$ and therefore for all $i\geq2$, $p_i \leq p_1\leq r/2$. Thus, every $w_i = p_i/r$ is smaller than $1/2$. By~\cref{lem_tailbound}
     we then have $H(w)\geq 1$ and therefore
     \begin{equation}
         H(p) = h(p_0) + rH(w)\geq h(p_0) + r \, .
     \end{equation}
     Since $h(x)\leq 1$, this finishes the proof.
\end{proof}

Let $\phi \in [0,\pi]$ and $U_{\phi}=\diag(1,1,1,\ee^{\ci \phi})$ be a controlled-phase gate. Any gate teleportation protocol for $U_{\phi}$ with a resource state $\ket{\psi}_{A'B'}$ of Schmidt rank $D\geq 3$ must fulfill
\begin{align}
H(A')_\psi 
\overset{\textnormal{\Cshref{thm_tightestconverse}}}{\geq} \min_{p_0 \in [1/2,\frac{1}{1+\sin(\phi/2)}]} \Big \{  h(p_0) + (1-p_0)h\left(\frac{\tau_D(p_0)}{1-p_0}\right) \Big \} \, . \label{eq_stronger_converse}
\end{align}
Note that for a resource state with $p_0 <1/2,$~\cref{lem_tailbound} yields $H(p) \geq 1,$ which is clearly bigger than the right-hand side of~\cref{eq_stronger_converse}. 
%%%%%%%%%%%%%%%%%%%%%%%%%%%%%%%%%%%%%%%%%%%%%%%%%%%%%%%%%%%%%%%%%%%%%%%%%%%%%%%%%%%
%%%%%%%%%%%%%%%%%%%%%%%%%%%%%%%%%%%%%%%%%%%%%%%%%%%%%%%%%%%%%%%%%%%%%%%%%%%%%%%%%%%
\bibliographystyle{arxiv_no_month}
\bibliography{bibliofile}

\end{document}

%% file: gate_teleporation_scheme.tex
\begin{tikzpicture}
\def \x{0.7};
\def \xb{1.2}
\def \y{0.4};
\definecolor{blue2}{rgb}{0.2, 0.2, 0.6}

\draw[] (-0.5*\x-0.5,0) -- (2*\x,0);
\draw[] (1.5*\x,-\y) -- (2*\x,-\y);
\draw[] (1.5*\x,-2*\y) -- (2*\x,-2*\y);
\draw[] (-0.5*\x-0.5,-3*\y) -- (2*\x,-3*\y);
\draw [] (2*\x,0.5*\y) rectangle (2*\x+\xb,-3.5*\y);
\node at (2*\x+0.5*\xb,-1.5*\y) {\textnormal{LOCC}};
\node at (1.5*\x-0.6-0.1,-1.5*\y) {$\ket{\psi}_{A'B'}$};
\draw[decorate,decoration={zigzag, segment length=4, amplitude=0.9}] (1.5*\x,-\y)  -- (1.5*\x,-2*\y);
\fill[black] (1.5*\x,-\y) circle (0.3mm);
\fill[black] (1.5*\x,-2*\y) circle (0.3mm);
\draw[] (2*\x+\xb,0) -- (3*\x+\xb,0);
\draw[] (2*\x+\xb,-3*\y) -- (3*\x+\xb,-3*\y);
\draw [dashed,blue2,thick,] (-0.15-0.5,0.5*\y+0.1) rectangle (3*\x+\xb-0.3,-3.5*\y-0.1);
\node[blue2] at (1.5*\x+0.5*\xb-0.2-0.25,-3.5*\y-0.35) {$U_{AB}$};
\node at (-0.5*\x-0.2-0.5,0) {$A$};
\node at (3*\x+\xb+0.2,0) {$A$};
\node at (-0.5*\x-0.2-0.5,-3*\y) {$B$};
\node at (3*\x+\xb+0.2,-3*\y) {$B$};

\end{tikzpicture}

%% file: new_gate_teleportation_scheme.tex
\begin{tikzpicture}
\def \x{0.7};
\def \xb{0.8}
\def \xbb{1}
\def \xbbb{1.2}
\def \y{0.4};

\definecolor{blue2}{rgb}{0.2, 0.2, 0.6}

\draw[fill=gray!10,gray!10] (-1.8*\x-0.5,2*\y+0.1) rectangle (11.5*\x+\xbbb+0.15,-2*\y+0.05);
\draw[fill=ForestGreen!10,ForestGreen!10] (-1.8*\x-0.5,-5*\y-0.1) rectangle (11.5*\x+\xbbb+0.15,-2*\y-0.05);

\draw[] (-1*\x-0.1-0.5,0) -- (2*\x,0);
\node at (-1*\x-0.3-0.5,0) {$A$};
\draw[] (\x,\y) -- (2*\x,\y);
\node at (\x-0.35,\y) {$\ket{0}_{\bar A}$};
\draw[] (\x,-\y) -- (2*\x,-\y);
\draw[] (\x,-3*\y) -- (4.3*\x,-3*\y);
\draw[decorate,decoration={zigzag, segment length=4, amplitude=0.9}] (\x,-\y)  -- (\x,-3*\y);
\node at (\x-0.65-0.1,-2*\y) {$\ket{\psi_{\phi}}_{A'B'}$};
\fill[black] (\x,-\y) circle (0.3mm);
\fill[black] (\x,-3*\y) circle (0.3mm);
\draw[] (-1*\x-0.1-0.5,-4*\y) -- (4.3*\x,-4*\y);
\node at (-1*\x-0.3-0.5,-4*\y) {$B$};

\draw [] (2*\x,1.5*\y) rectangle (2*\x+\xb,-1.5*\y);
\node at (2*\x+0.5*\xb,0) {$U_{\phi,1}$};
\node[meter, scale=0.42] at (2*\x+\xb+0.4,\y) {};
\draw [] (2*\x+\xb,\y) -- (2*\x+\xb+0.18,\y);

\draw [] (4.3*\x,-2.5*\y) rectangle (4.3*\x+\xb,-4.5*\y);
\node at (4.3*\x+0.5*\xb,-3.5*\y) {$U^{(x)}_{\phi,2}$};
\draw[densely dotted] (2*\x+\xb+0.63,\y) -- (4.3*\x+0.5*\xb,\y);
\draw[densely dotted] (4.3*\x+0.5*\xb,\y) -- (4.3*\x+0.5*\xb,-2.5*\y);
\node at (2*\x+\xb+0.9,\y+0.15) {\small{$x$}};
\node[meter, scale=0.42] at (4.3*\x+\xb+0.4,-3*\y) {};
\draw (4.3*\x+\xb,-3*\y) -- (4.3*\x+\xb+0.18,-3*\y);

\draw [] (7*\x,0.5*\y) rectangle (7*\x+\xbb,-1.5*\y);
\node at (7*\x+0.5*\xbb,-0.5*\y) {$U^{(x,y)}_{\phi,3}$};
\draw[] (2*\x+\xb,0) -- (7*\x,0);
\draw[] (2*\x+\xb,-\y) -- (7*\x,-\y);
\draw[densely dotted] (4.3*\x+\xb+0.63,-3*\y) -- (7*\x+0.5*\xbb,-3*\y);
\draw[densely dotted] (7*\x+0.5*\xbb,-3*\y) -- (7*\x+0.5*\xbb,-1.5*\y);
\node at (7*\x+0.5*\xbb-0.7,-3*\y+0.14) {\small{$y$}};
\node[meter, scale=0.42] at (7*\x+\xbb+0.4,-1*\y) {};
\draw (7*\x+\xbb,-1*\y) -- (7*\x+\xbb+0.18,-1*\y);

\draw [] (10*\x,-3.5*\y+0.1) rectangle (10*\x+\xbbb,-4.5*\y-0.1);
\node at (10*\x+0.5*\xbbb,-4*\y) {$U^{(x,z)}_{\phi,4}$};
\draw[] (4.3*\x+\xb,-4*\y) -- (10*\x,-4*\y);
\draw[densely dotted] (7*\x+\xbb+0.63,-1*\y) -- (10*\x+0.5*\xbbb,-1*\y);
\draw[densely dotted] (10*\x+0.5*\xbbb,-1*\y) -- (10*\x+0.5*\xbbb,-3.5*\y+0.1);
\node at (7*\x+\xbb+0.63+0.25,-1*\y+0.12) {\small{$z$}};
\draw[] (10*\x+\xbbb,-4*\y) -- (11*\x+\xbbb,-4*\y);
\draw[] (7*\x+\xbb,0) -- (11*\x+\xbbb,0);
\node at (11*\x+\xbbb+0.15,0) {$A$};
\node at (11*\x+\xbbb+0.15,-4*\y) {$B$};

\draw [blue2,dashed,thick] (-0.75*\x-0.1-0.4,2*\y) rectangle (10.5*\x+\xbbb+0.1,-5*\y);
\node[blue2] at (5.625*\x+0.5*\xbb-0.2,-5*\y-0.5) {$U_{\phi}$};

\end{tikzpicture}

%% file: Fig_results.tex
\begin{tikzpicture}
\definecolor{blue2}{rgb}{0.2, 0.2, 0.6}
\begin{axis}[
  width=0.82\textwidth,height=6.0cm,
  xmin=0,xmax=1,ymin=0,ymax=1.04,
  xlabel={Normalized phase $\phi/\pi$},
  ylabel={Entanglement entropy (ebits)},
  grid=major,grid style={gray!18},
  legend cell align=left,
  legend style={draw=none,fill=none,at={(0.98,0.06)},anchor=south east},
  samples=240,domain=0.00001:1
]

\addplot[name path=route,thick,black]
 {min(1,-((1/(1+sin(deg(pi*x/2))))*ln(1/(1+sin(deg(pi*x/2))))
 +(sin(deg(pi*x/2))/(1+sin(deg(pi*x/2))))
 *ln(sin(deg(pi*x/2))/(1+sin(deg(pi*x/2)))))/ln(2)
 +sin(deg(pi*x/2))/(1+sin(deg(pi*x/2))))};
\addlegendentry{achievability bound (see~\cref{eq_achievablity}) }
\addplot[name path=base,thick,black,dashed]
 {-((1/(1+sin(deg(pi*x/2))))*ln(1/(1+sin(deg(pi*x/2))))
 +(sin(deg(pi*x/2))/(1+sin(deg(pi*x/2))))
 *ln(sin(deg(pi*x/2))/(1+sin(deg(pi*x/2)))))/ln(2)};
\addlegendentry{converse bound (see~\cref{eq_converse})}
\addplot[thick,dotted,blue2]
 table[col sep=comma,x=phi_over_pi,
   y=D3_entropy_lower_bound_ebits]{improved_converse.csv};
\addlegendentry{converse with Schmidt rank $3$ (see \cref{app_improved_converse})}
%\addplot[gray!11] fill between[of=base and route];
\end{axis}
\end{tikzpicture}

%% file: Teleportation_correctness.tex
\begin{tikzpicture}
\def \x{0.7};
\def \xb{0.8}
\def \xbb{1}
\def \xbbb{1.2}
\def \y{0.4};

\definecolor{blue2}{rgb}{0.2, 0.2, 0.6}

% \draw[fill=gray!10,gray!10] (-1.8*\x,2*\y+0.1) rectangle (11.5*\x+\xbbb+0.15,-2*\y+0.05);
% \draw[fill=ForestGreen!10,ForestGreen!10] (-1.8*\x,-5*\y-0.1) rectangle (11.5*\x+\xbbb+0.15,-2*\y-0.05);

\draw[] (-1*\x-0.1,0) -- (2*\x,0);
\node at (-1*\x-0.3,0) {$A$};
\draw[] (\x,\y) -- (2*\x,\y);
\node at (\x-0.35,\y) {$\ket{0}_{\bar A}$};
\draw[] (\x,-\y) -- (2*\x,-\y);
\draw[] (\x,-3*\y) -- (4.3*\x,-3*\y);
\draw[decorate,decoration={zigzag, segment length=4, amplitude=0.9}] (\x,-\y)  -- (\x,-3*\y);
\node at (\x-0.65,-2*\y) {$\ket{\psi_{\phi}}_{A'B'}$};
\fill[black] (\x,-\y) circle (0.3mm);
\fill[black] (\x,-3*\y) circle (0.3mm);
\draw[] (-1*\x-0.1,-4*\y) -- (4.3*\x,-4*\y);
\node at (-1*\x-0.3,-4*\y) {$B$};

\draw [] (2*\x,1.5*\y) rectangle (2*\x+\xb,-1.5*\y);
\node at (2*\x+0.5*\xb,0) {$U_{\phi,1}$};
\node[meter, scale=0.42] at (2*\x+\xb+0.4,\y) {};
\draw [] (2*\x+\xb,\y) -- (2*\x+\xb+0.18,\y);

\draw [] (4.3*\x,-2.5*\y) rectangle (4.3*\x+\xb,-4.5*\y);
\node at (4.3*\x+0.5*\xb,-3.5*\y) {$U^{(x)}_{\phi,2}$};
\draw[densely dotted] (2*\x+\xb+0.63,\y) -- (4.3*\x+0.5*\xb,\y);
\draw[densely dotted] (4.3*\x+0.5*\xb,\y) -- (4.3*\x+0.5*\xb,-2.5*\y);
\node at (2*\x+\xb+0.9,\y+0.15) {\small{$x$}};
\node[meter, scale=0.42] at (4.3*\x+\xb+0.4,-3*\y) {};
\draw (4.3*\x+\xb,-3*\y) -- (4.3*\x+\xb+0.18,-3*\y);

\draw [] (7*\x,0.5*\y) rectangle (7*\x+\xbb,-1.5*\y);
\node at (7*\x+0.5*\xbb,-0.5*\y) {$U^{(x,y)}_{\phi,3}$};
\draw[] (2*\x+\xb,0) -- (7*\x,0);
\draw[] (2*\x+\xb,-\y) -- (7*\x,-\y);
\draw[densely dotted] (4.3*\x+\xb+0.63,-3*\y) -- (7*\x+0.5*\xbb,-3*\y);
\draw[densely dotted] (7*\x+0.5*\xbb,-3*\y) -- (7*\x+0.5*\xbb,-1.5*\y);
\node at (7*\x+0.5*\xbb-0.7,-3*\y+0.14) {\small{$y$}};
\node[meter, scale=0.42] at (7*\x+\xbb+0.4,-1*\y) {};
\draw (7*\x+\xbb,-1*\y) -- (7*\x+\xbb+0.18,-1*\y);

\draw [] (10*\x,-3.5*\y+0.1) rectangle (10*\x+\xbbb,-4.5*\y-0.1);
\node at (10*\x+0.5*\xbbb,-4*\y) {$U^{(x,z)}_{\phi,4}$};
\draw[] (4.3*\x+\xb,-4*\y) -- (10*\x,-4*\y);
\draw[densely dotted] (7*\x+\xbb+0.63,-1*\y) -- (10*\x+0.5*\xbbb,-1*\y);
\draw[densely dotted] (10*\x+0.5*\xbbb,-1*\y) -- (10*\x+0.5*\xbbb,-3.5*\y+0.1);
\node at (7*\x+\xbb+0.63+0.25,-1*\y+0.12) {\small{$z$}};
\draw[] (10*\x+\xbbb,-4*\y) -- (11*\x+\xbbb,-4*\y);
\draw[] (7*\x+\xbb,0) -- (11*\x+\xbbb,0);
\node at (11*\x+\xbbb+0.15,0) {$A$};
\node at (11*\x+\xbbb+0.15,-4*\y) {$B$};

% \draw [blue2,dashed,thick] (-0.75*\x-0.1,2*\y) rectangle (10.5*\x+\xbbb+0.1,-5*\y);
% \node[blue2] at (5.625*\x+0.5*\xbb,-5*\y-0.5) {$U_{\phi}$};

 \draw [blue2,dashed,thick] (-0.7,1.5*\y+0.1) -- (-0.7,-5*\y-0.1);
\node[blue2] at (-0.7,-5*\y-0.4) {$\ket{\Lambda}_{AB}$};

 \draw [blue2,dashed,thick] (4*\x+0.1,1.5*\y+0.1) -- (4*\x+0.1,-5*\y-0.1);
\node[blue2] at (4*\x-0.3,-5*\y-0.4) {$\ket{\Omega_1^{(x)}}_{AA'BB'}$};

 \draw [blue2,dashed,thick] (6.3*\x+0.1,1.5*\y+0.1) -- (6.3*\x+0.1,-5*\y-0.1);
\node[blue2] at (6.3*\x+0.4,-5*\y-0.4) {$\ket{\Omega_2^{(x,y)}}_{AA'B}$};

 \draw [blue2,dashed,thick] (9.3*\x+0.1,1.5*\y+0.1) -- (9.3*\x+0.1,-5*\y-0.1);
\node[blue2] at (9.3*\x+0.4,-5*\y-0.4) {$\ket{\Omega_3^{(x,y,z)}}_{AB}$};

 \draw [blue2,dashed,thick] (12.2*\x+0.1,1.5*\y+0.1) -- (12.2*\x+0.1,-5*\y-0.1);
\node[blue2] at (12.2*\x+0.6,-5*\y-0.4) {$\ket{\Omega_4^{(x,y,z)}}_{AB}$};

\end{tikzpicture}

%% file: Fig_results_stronger_converse.tex
\begin{tikzpicture}
\definecolor{blue2}{rgb}{0.2, 0.2, 0.6}
\begin{axis}[
  width=0.90\textwidth,height=6.0cm,
  xmin=0,xmax=1,ymin=0,ymax=1.04,
  xlabel={Normalized phase $\phi/\pi$},
  ylabel={Entanglement entropy (ebits)},
  grid=major,grid style={gray!18},
  legend cell align=left,
legend style={ draw=none, fill=white, fill opacity=0.8, text opacity=1, at={(0.97,0.03)}, anchor=south east, font=\small }, legend cell align=left, legend columns=2,
  samples=240,domain=0.00001:1
]

\addplot[name path=route,thick,black]
 table[
   col sep=comma,
   x=phi_over_pi,
   y expr={min(1,\thisrow{h2_q_plus_1_minus_q_ebits})}
 ]{improved_converse.csv};
\addlegendentry{achievability (see~\cref{eq_achievablity})}

% The D=infinity column is the formal D -> infinity limit of Theorem B.1,
% which coincides with the elementary converse h(q_phi).
% \addplot[name path=base,draw=none,forget plot]
%  table[col sep=comma,x=phi_over_pi,
%    y=D_inf_unrestricted_entropy_lower_bound_ebits]{improved_converse.csv};
%\addplot[gray!11,forget plot] fill between[of=base and route];

\addplot[thick,dotted,blue2]
 table[col sep=comma,x=phi_over_pi,
   y=D3_entropy_lower_bound_ebits]{improved_converse.csv};
\addlegendentry{converse, $D=3$}
\addplot[thick,cyan,dash dot]
 table[col sep=comma,x=phi_over_pi,
   y=D4_entropy_lower_bound_ebits]{improved_converse.csv};
\addlegendentry{converse, $D=4$}
\addplot[thick,ForestGreen,dash pattern=on 5pt off 1.5pt on 1pt off 1.5pt]
 table[col sep=comma,x=phi_over_pi,
   y=D10_entropy_lower_bound_ebits]{improved_converse.csv};
\addlegendentry{converse, $D=10$}
\addplot[thick,black,dashed]
 table[col sep=comma,x=phi_over_pi,
   y=D_inf_unrestricted_entropy_lower_bound_ebits]{improved_converse.csv};
% \addplot[name path=base,thick,black,dashed]
%  {-((1/(1+sin(deg(pi*x/2))))*ln(1/(1+sin(deg(pi*x/2))))
%  +(sin(deg(pi*x/2))/(1+sin(deg(pi*x/2))))
%  *ln(sin(deg(pi*x/2))/(1+sin(deg(pi*x/2)))))/ln(2)};
\addlegendentry{converse, $D=\infty$ }
\end{axis}
\end{tikzpicture}